\documentclass[11pt,letterpaper]{article}
\usepackage{matharticle}
\newcommand{\e}{\eps}
\DeclareMathOperator{\pE}{\mathop{\tilde{\mathbb{E}}}}

\newcommand{\Paren}[1]{\left(#1\right)}

\newcommand{\Brac}[1]{\left[#1\right]}

\newcommand{\Abs}[1]{\left\lvert#1\right\rvert}

\newcommand{\Norm}[1]{\left\lVert#1\right\rVert}

\newcommand{\iprod}[1]{\langle#1\rangle}
\newcommand{\Iprod}[1]{\left\langle#1\right\rangle}

\newcommand{\mper}{\,.}
\newcommand{\mcom}{\,,}

\newtheorem{problem}[theorem]{Problem}

\newcommand{\calB}{\mathcal B}

\newcommand{\proves}[1]{\vdash_{#1} \, }
\newcommand{\tensor}{\otimes}

\newcommand*{\Id}{\mathrm{Id}}

\newcommand{\cA}{\mathcal A}

\newcommand{\cD}{\mathcal D}

\newcommand{\cL}{\mathcal L}

\newcommand{\cN}{\mathcal N}

\newcommand{\cX}{\mathcal X}

\newcommand{\cAhat}{\widehat{\cA}}

\newcommand{\bmu}{\overline{\mu}}

\newcommand\MYcurrentlabel{xxx}
\newcommand{\MYstore}[2]{%
  \global\expandafter \def \csname MYMEMORY #1 \endcsname{#2}%
}
\newcommand{\MYload}[1]{%
  \csname MYMEMORY #1 \endcsname%
}
\newcommand{\MYnewlabel}[1]{%
  \renewcommand\MYcurrentlabel{#1}%
  \MYoldlabel{#1}%
}
\newcommand{\MYdummylabel}[1]{}
\newcommand{\torestate}[1]{%
  \let\MYoldlabel\label%
  \let\label\MYnewlabel%
  #1%
  \MYstore{\MYcurrentlabel}{#1}%
  \let\label\MYoldlabel%
}
\newcommand{\restatetheorem}[1]{%
  \let\MYoldlabel\label
  \let\label\MYdummylabel
  \begin{theorem*}[Restatement of \cref{#1}]
    \MYload{#1}
  \end{theorem*}
  \let\label\MYoldlabel
}
\newcommand{\restatelemma}[1]{%
  \let\MYoldlabel\label
  \let\label\MYdummylabel
  \begin{lemma*}[Restatement of \cref{#1}]
    \MYload{#1}
  \end{lemma*}
  \let\label\MYoldlabel
}
\newcommand{\restateprop}[1]{%
  \let\MYoldlabel\label
  \let\label\MYdummylabel
  \begin{proposition*}[Restatement of \cref{#1}]
    \MYload{#1}
  \end{proposition*}
  \let\label\MYoldlabel
}
\newcommand{\restatefact}[1]{%
  \let\MYoldlabel\label
  \let\label\MYdummylabel
  \begin{fact*}[Restatement of \prettyref{#1}]
    \MYload{#1}
  \end{fact*}
  \let\label\MYoldlabel
}
\newcommand{\restate}[1]{%
  \let\MYoldlabel\label
  \let\label\MYdummylabel
  \MYload{#1}
  \let\label\MYoldlabel
}

\usepackage[
pagebackref,
colorlinks=true,
urlcolor=blue,
linkcolor=blue,
citecolor=OliveGreen,
]{hyperref}
\usepackage{paralist}

\title{Mixture Models, Robustness, and Sum of Squares Proofs}

\author{
    Samuel B. Hopkins\thanks{Supported by an NSF graduate research fellowship, a Microsoft Research PhD fellowship, a Cornell University fellowship, and David Steurer's NSF CAREER award. Part of this work was accomplished while this author was an intern at Microsoft Research New England.}\\
    Cornell University\\
    \texttt{samhop{@}cs.cornell.edu} 
    \and
    Jerry Li\thanks{Supported by NSF CAREER Award CCF-1453261, CCF-1565235, a Google Faculty Research Award, and an NSF Graduate Research Fellowship.}\\
    MIT\\
    \texttt{jerryzli{@}mit.edu}
}

\begin{document}
\maketitle
\thispagestyle{empty} 

\begin{abstract}
    We use the Sum of Squares method to develop new efficient algorithms for learning well-separated mixtures of Gaussians and robust mean estimation, both in high dimensions, that substantially improve upon the statistical guarantees achieved by previous efficient algorithms.
    Our contributions are:
    \begin{itemize}
    \item \textbf{Mixture models with separated means: } We study mixtures of $k$ distributions in $d$ dimensions, where the means of every pair of distributions are separated by at least $k^{\e}$.
    In the special case of spherical Gaussian mixtures, we give a $(dk)^{O(1/\e^2)}$-time algorithm that learns the means assuming separation at least $k^{\eps}$, for any $\eps > 0$.
    This is the first algorithm to improve on greedy (``single-linkage'') and spectral clustering, breaking a long-standing barrier for efficient algorithms at separation $k^{1/4}$.

    \item \textbf{Robust estimation: } When an unknown $(1-\e)$-fraction of $X_1,\ldots,X_n$ are chosen from a sub-Gaussian distribution with mean $\mu$ but the remaining points are chosen adversarially, we give an algorithm recovering $\mu$ to error $\e^{1-1/t}$ in time $d^{O(t^2)}$, so long as sub-Gaussian-ness up to $O(t)$ moments can be certified by a Sum of Squares proof.
    This is the first polynomial-time algorithm with guarantees approaching the information-theoretic limit for non-Gaussian distributions.
    Previous algorithms could not achieve error better than $\eps^{1/2}$.
    \end{itemize}
    Both of these results are based on a unified technique.
    Inspired by recent algorithms of Diakonikolas et al. in robust statistics, we devise an SDP based on the Sum of Squares method for the following setting: given $X_1,\ldots,X_n \in \R^d$ for large $d$ and $n = \poly(d)$ with the promise that a subset of $X_1,\ldots,X_n$ were sampled from a probability distribution with bounded moments, recover some information about that distribution.
\end{abstract}

\newpage
\thispagestyle{empty} 

\tableofcontents{}

\newpage
\setcounter{page}{1}


\section{Introduction}
We propose and analyze a family of new algorithms for some fundamental high-dimensional statistical estimation problems.
In particular, we give new algorithms for the following problems.
\begin{enumerate}
  \item \textbf{Learning $\Delta$-separated mixture models:} Given $n$ samples $X_1,\ldots,X_n \in \R^d$ from a mixture of $k$ probability distributions $\cD_1,\ldots,\cD_k$ on $\R^d$ with means $\mu_1,\ldots,\mu_k \in \R^d$ and covariances $\Sigma_1,\ldots,\Sigma_k \preceq \Id$, where $\|\mu_i - \mu_j\| \geq \Delta$, estimate $\mu_1,\ldots,\mu_k$.\footnote{A mixture model consists of probability distributions $\cD_1,\ldots,\cD_k$ on $\R^d$ and mixing weights $\lambda_1,\ldots,\lambda_k \geq 0$ with $\sum_{i \leq k} \lambda_i = 1$. The distribution $\cD_i$ has mean $\mu_i$. Each sample $x_j$ is generated by first sampling a component $i \in [k]$ according to the weights $\lambda$, then sampling $x_j \sim \cD_i$.}
  \item \textbf{Robust mean estimation:} Given $n$ vectors $X_1,\ldots,X_n \in \R^d$, of which a $(1-\e)$-fraction are samples from a probability distribution $\cD$ with mean $\mu$ and covariance $\Sigma \preceq \Id$ and the remaining $\e$-fraction are arbitrary vectors (which may depend on the $(1-\e)n$ samples from $\cD$), estimate $\mu$.
\end{enumerate}
Mixture models, and especially Gaussian mixture models (where $\cD_1,\ldots,\cD_k$ are Gaussian distributions) have been studied since Pearson in 1894 \cite{pearson1894contributions}.
Work in theoretical computer science dates at least to the pioneering algorithm of Dasgupta in 1999 \cite{dasgupta1999learning}, which has been followed by numerous other algorithms and lower bounds \cite{wu1983convergence,dasgupta2007probabilistic,MR2115036-Arora05, DBLP:conf/focs/VempalaW02, DBLP:conf/focs/KumarK10, achlioptas2005spectral, DBLP:conf/colt/FeldmanSO06,DBLP:conf/stoc/KalaiMV10,DBLP:conf/focs/BelkinS10,DBLP:conf/focs/MoitraV10,MR3385380-Hsu13, DBLP:conf/colt/AndersonBGRV14, DBLP:conf/stoc/BhaskaraCMV14,DK14, SOAJ14,DBLP:conf/stoc/HardtP15,xu2016global,MR3388256-Ge15,LS17,regev-vijayaraghavan-17,daskalakis2016ten}.

Robust estimation in the form we study here is a more recent transplant to theoretical computer science \cite{DBLP:conf/focs/DiakonikolasKK016, DBLP:conf/focs/LaiRV16,DBLP:journals/corr/CharikarSV16,DKS16c,cherapanamjeri2017thresholding,DKKLMS17,DKKLMS17b,SCV17}, but statisticians have long sought outlier-robust estimators.
Formal study of arbitrarily-bad/adversarially-chosen outliers originates in the 1960s with the advent of ``breakdown points'' in statistics \cite{huber1964robust,Tuk75, HRRS86,JP78,Ber06}.

Though outwardly rather different, mixture model learning and robust estimation share some underlying structure.
An algorithm for either must identify or otherwise recover information about one or several \emph{structured} subsets of a number of samples $X_1,\ldots,X_n \in \R^d$.
In the mixture model case, each collection of all the samples from each distribution $\cD_i$ is a structured subset.
In the robust estimation case there is just one structured subset: the $(1-\e)n$ samples drawn from the distribution $\cD$.\footnote{The recent work \cite{DBLP:journals/corr/CharikarSV16} codifies this similarity by unifying both these problems into what they call a list-decodable learning setting.}
Our algorithms are based on new techniques for identifying such structured subsets of points in large data sets.

For mixture models, a special case of our main result yields the first progress in more than 15 years on efficiently clustering mixtures of separated spherical Gaussians.
The question here is: if $\cD_1,\ldots,\cD_k$ are all Gaussian with covariance identity and $k = \poly(d)$, what is the minimum cluster separation $\Delta$ which allows for a polynomial-time algorithm to estimate $\mu_1,\ldots,\mu_k$ from $\poly(k,d)$ samples from the mixture model?
The guarantees of the previous best algorithms for this problem, which require $\Delta \geq O(k^{1/4})$, are captured by a simple greedy clustering algorithm, sometimes called \emph{single-linkage clustering}: when $\Delta \geq O(k^{1/4})$, with high probability every pair of samples from the same cluster is closer in Euclidean distance than every pair of samples from differing clusters.
We break this single-linkage clustering barrier: for every $\gamma > 0$ we give a $\poly(k,d)$-time algorithm for this problem when $\Delta > k^\gamma$.

To do so we make novel algorithmic use of higher moments (in fact, $O(1/\gamma)$ moments) of the underlying distributions $\cD_i$.
Our main technical contribution is a new algorithmic technique for finding either a structured subset of data points or the empirical mean of such a subset when the subset consists of independent samples from a distribution $\cD$ which has bounded higher-order moments \emph{and there is a simple certificate of this boundedness}.
This technique leverages the Sum of Squares (SoS) hierarchy of semidefinite programs (SDPs), and in particular a powerful approach for designing SoS-based algorithms in machine learning settings, developed and used in \cite{DBLP:conf/stoc/BarakKS14, DBLP:conf/stoc/BarakKS15, ge2015decomposing, DBLP:conf/colt/BarakM16,DBLP:conf/colt/HopkinsSS15,DBLP:conf/focs/MaSS16,DBLP:journals/corr/PotechinS17}.
We suspect use of higher moments is necessary in light of second-moment indistinguishability results for mixtures with small separation \cite{achlioptas2005spectral}.

This SoS approach to unsupervised learning rests on a notion of \emph{simple identifiability proofs:} the main step in designing an algorithm using SoS to recover some parameters $\theta$ from samples $x_1,\ldots,x_n \sim p(x \, | \, \theta)$ is to prove in a restricted proof system that $\theta$ is likely to be uniquely identifiable from $x_1,\ldots,x_n$.
We develop this thoroughly later on, but roughly speaking one may think of this as requiring the identifiability proof to use only simple inequalities, such as Cauchy-Schwarz and H\"older's inequality, applied to low-degree polynomials.
The simple identifiability proofs we construct for both the mixture models and robust estimation settings are heavily inspired by the robust estimation algorithms of Diakonikolas et al. \cite{DBLP:conf/focs/DiakonikolasKK016}.

\subsection{Results}
Both of the problems we study have a long history; for now we just note some highlights and state our main results.

\paragraph{Mixture models} The problem of learning mixture models dates to Pearson in 1894, who invented the method of moments in order to separate a mixture of two Gaussians \cite{pearson1894contributions}.
Mixture models have since become ubiquitous in data analysis across many disciplines \cite{titterington1985statistical,mclachlan2004finite}.
In recent years, computer scientists have devised many ingenious algorithms for learning mixture models as it became clear that classical statistical methods (e.g. maximum likelihood estimation) often suffer from computational intractability, especially when there are many mixture components or the components are high dimensional.

A highlight of this work is a series of algorithmic results when the components of the mixture model are Gaussian \cite{dasgupta1999learning,dasgupta2007probabilistic,MR2115036-Arora05, DBLP:conf/focs/VempalaW02}.
Here the main question is: how small can the cluster separation $\Delta$ be such that there exists an algorithm to estimate $\mu_1,\ldots,\mu_k$ from samples $x_1,\ldots,x_n$ in $\poly(k,d)$ time (hence also using $n = \poly(k,d)$ samples)?
Focusing for simplicity on spherical Gaussian components (i.e. with covariance equal to the identity matrix $\Id$) and with number of components similar to the ambient dimension of the data (i.e. $k = d$) and uniform mixing weights (i.e. every cluster has roughly the same representation among the samples), the best result in previous work gives a $\poly(k)$-time algorithm when $\Delta \geq k^{1/4}$.

Separation $\Delta = k^{1/4}$ represents a natural algorithmic barrier: when $\Delta \geq k^{1/4}$, \emph{every pair of samples from the same cluster are closer to each other in Euclidean distance than are every pair of samples from distinct clusters (with high probability)}, while this is no longer true if $\Delta < k^{1/4}$.
Thus, when $\Delta \geq k^{1/4}$, a simple greedy algorithm correctly clusters the samples into their components (this algorithm is sometimes called \emph{single-linkage clustering}).
On the other hand, standard information-theoretic arguments show that the means remain approximately identifiable from $\poly(k,d)$ samples when $\Delta$ is as small as $O(\sqrt{\log k})$, but these methods yield only exponential-time algorithms.\footnote{Recent and sophisticated arguments show that the means are identifiable (albeit inefficiently) with error depending only on the number of samples and not on the separation $\Delta$ even when $\Delta = O(\sqrt{\log k})$ \cite{regev-vijayaraghavan-17}.}
Nonetheless, despite substantial attention, this $\Delta = k^{1/4}$ barrier representing the breakdown of single-linkage clustering has stood for nearly 20 years.

We prove the following main theorem, breaking the single-linkage clustering barrier.
\begin{theorem}[Informal, special case for uniform mixture of spherical Gaussians]

\label{thm:mixtures-intro}
  For every $\gamma > 0$ there is an algorithm with running time $(dk)^{O(1/\gamma^2)}$ using at most $n \leq k^{O(1)} d^{O(1/\gamma)}$ samples which, given samples $x_1,\ldots,x_n$ from a uniform mixture of $k$ spherical Gaussians $\cN(\mu_i,\Id)$ in $d$ dimensions with means $\mu_1,\ldots,\mu_k \in \R^d$ satisfying $\|\mu_i - \mu_j\| \geq k^{\gamma}$ for each $i \neq j$, returns estimators $\hat{\mu}_1,\ldots,\hat{\mu}_k \in \R^d$ such that $\|\hat{\mu}_i - \mu_i\| \leq 1/\poly(k)$ (with high probability).
\end{theorem}
We pause here to make several remarks about this theorem.
Our algorithm makes novel use of higher order moments of Gaussian (and sub-Gaussian) distributions.
Previous work for efficiently learning well-separated mixtures of Gaussians used only second order moment information, whereas we use $O(1 / \gamma)$ moments.

The guarantees of this theorem hold well beyond the Gaussian setting; the theorem applies to any mixture model with $k^{\gamma}$ separation and whose component distributions $\cD_1,\ldots,\cD_k$ are what we term $O(1/\gamma)$-\emph{explicitly bounded}.
We define this notion formally below, but roughly speaking, a $t$-explicitly bounded distribution $\cD$ has $t$-th moments obeying a subgaussian-type bound---that is, for every unit vector $u \in \R^d$ one has $\E_{Y \sim \cD} |\iprod{Y,u}|^t \leq t^{t / 2}$---and there is a certain kind of \emph{simple certificate} of this fact, namely a low-degree Sum of Squares proof.
Among other things, this means the theorem also applies to mixtures of symmetric product distributions with bounded moments.

For mixtures of distributions with sufficiently-many bounded moments (such as Gaussians), our guarantees go even further.
We show that using $d^{O(\log k)^2}$ time and $d^{O(\log k)}$ samples, we can recover the means to error $1 / \poly (k)$ even if the separation is only $C \sqrt{\log k}$ for some universal constant $C$.
Strikingly, \cite{regev-vijayaraghavan-17} show that any algorithm that can learn the means nontrivially given separation $o (\sqrt{\log k})$ must require super-polynomial samples and time.
Our results show that just above this threshold, it is possible to learn with just quasipolynomially many samples and time.

Finally, throughout the paper we state error guarantees roughly in terms of obtaining $\hat{\mu_i}$ with $\|\hat{\mu}_i - \mu_i\| \leq 1/\poly(k) \ll k^{\gamma}$, meaning that we get $\ell_2$ error which is much less than the true separation.
In the special case of spherical Gaussians, we note that we can use our algorithm as a warm-start to recent algorithms due to \cite{regev-vijayaraghavan-17}, and achieve error $\delta$ using $\poly (m, k, 1 / \delta)$ additional runtime and samples for some polynomial independent of $\gamma$.

\paragraph{Robust mean estimation}
Estimators which are robust to outlying or corrupted samples have been studied in statistics at least since the 1960s \cite{huber1964robust, tukey1975mathematics}.
The model we consider in this paper is a slight generalization of H\"{u}ber's contamination model \cite{huber1964robust}.
We are given $X_1, \ldots, X_n$, originally drawn iid from some unknown distribution $\cD$, but an adversary has changed an $\eps$ fraction of these points adversarially.
We call such a set of points $\eps$-corrupted.\footnote{H\"{u}ber's contamination model essentially only allows the adversary to add corrupted points, but not remove uncorrupted points.}
The goal of robust statistics is to recover statistics of $\cD$ such as mean and covariance, given $\eps$-corrupted samples from $\cD$.

In classical robust statistics, the robust mean estimation problem is known as \emph{robust estimation of location}, and robust covariance estimation is known as \emph{robust estimation of scale}.
Classical works consider a measure known as breakdown point, which is (informally) the fraction of samples that an adversary must corrupt before the estimator has no provable guarantees.
They often design robust estimators for mean and covariance that achieve optimal error in many fundamental settings.
For instance, given samples from a symmetric sub-Gaussian distribution in $k$ dimensions such that an $\eps$-fraction are arbitrarily corrupted, an estimator known as the Tukey median \cite{tukey1975mathematics} achieves error $O(\eps)$, which is information theoretically optimal.
However, these estimators are all $NP$-hard to compute \cite{JP78,Ber06} and the best known algorithms require $\exp(d)$ time in general.

For a long time, all known computationally efficient robust statistics for the mean or covariance of a $d$-dimensional Gaussian had error degrading polynomially with the dimension.\footnote{We remark that this was the state of affairs even for the H\"{u}ber contamination model.}
In recent work, \cite{DBLP:conf/focs/DiakonikolasKK016, DBLP:conf/focs/LaiRV16} gave efficient and robust estimators for these statistics which achieve substantially better error.
In particular, \cite{DBLP:conf/focs/DiakonikolasKK016} achieve error $O(\eps \sqrt{\log 1 / \eps})$ for estimating the mean of a Gaussian with identity covariance, and error $O(\eps \log^{3/2} 1 / \eps)$ for robustly estimating the mean of a Gaussian with unknown variance $\Sigma \preceq I$.

Unfortunately, these results are somewhat tailored to Gaussian distributions, or require covariance very close to identity.
For general sub-Gaussian distributions with unknown variance $\Sigma \preceq I$, the best known efficient algorithms achieve only $O(\eps^{1/2})$ error \cite{DKKLMS17b, SCV17}.
We substantially improve this, under a slightly stronger condition than sub-Gaussianity.
Recall that a distribution $\cD$ with mean $\mu$ over $\R^d$ is sub-Gaussian if for every unit vector $u$ and every $t \in \N$ even, the following moment bound holds:
\[
\E_{X \sim \cD} \iprod{u, X - \mu}^t \leq t^{t / 2} \; .
\]
Informally stated, our algorithms will work under the condition that this moment bound can be certified by a low degree SoS proof, for all $s \leq t$.
We call such distributions $t$-\emph{explicitly bounded} (we are ignoring some parameters, see Definition~\ref{def:explicitly-bounded} for a formal definition).
This class captures many natural sub-Gaussian distributions, such as Gaussians, product distributions of sub-Gaussians, and rotations thereof (see Appendix~\ref{sec:explicitly-bounded-families}).
For such distributions, we show:

\begin{theorem}[informal, see Theorem~\ref{thm:robust-main}]
\label{thm:robust-intro}
	Fix $\eps > 0$ sufficiently small and let $t \geq 4$.
	Let $\cD$ be a $O(t)$-explicitly bounded distribution over $\R^d$ with mean $\mu^*$.
	There is an algorithm with sample complexity $d^{O(t)} (1/\eps)^{O(1)}$ running time $(d^t \eps)^{O(t)}$ such that given an $\eps$-corrupted set of samples of sufficiently large size from $\cD$, outputs $\mu$ so that with high probability $\| \mu - \mu^* \| \leq O(\eps^{1 - 1 / t})$.
\end{theorem}

As with mixture models, we can push our statistical rates further, if we are willing to tolerate quasipolynomial runtime and sample complexity.
In particular, we can obtain error $O(\eps \sqrt{\log 1 / \eps})$ error with $d^{O(\log 1 / \eps)}$ samples and $d^{O(\log 1/\eps)^2}$ time.

\subsection{Related work}

\paragraph{Mixture models}
The literature on mixture models is vast so we cannot attempt a full survey here.
The most directly related line of work to our results studies mixtures models under mean-separation conditions, and especially mixtures of Gaussians, where the number $k$ of components of the mixture grows with the dimension $d$ \cite{dasgupta1999learning,dasgupta2007probabilistic,MR2115036-Arora05, DBLP:conf/focs/VempalaW02}.
The culmination of these works is the algorithm of Vempala and Wang, which used spectral dimension reduction to improve on the $d^{1/4}$ separation required by previous works to $k^{1/4}$ in $\ell_2$ distance for $k \leq d$ spherical Gaussians in $d$ dimensions.

Other works have relaxed the requirement that the underlying distributions be Gaussian \cite{DBLP:conf/focs/KumarK10, achlioptas2005spectral}; we also address non-Gaussian distributions, relaxing the Gaussian-ness assumption to explicit moment boundedness.
One recent work in this spirit uses SDPs to cluster mixture models under separation assumptions \cite{mixon2017clustering}; the authors show that a standard SDP relaxation of $k$-means achieves guarantees comparable to previously-known specially-tailored mixture model algorithms; however, this algorithm suffers from the same $k^{1/4}$ barrier as other previous works.

\emph{Sample complexity: }
Recent work of \cite{regev-vijayaraghavan-17} considers the Gaussian mixtures problem in an information-theoretic setting: they show that there is some constant $C$ so that if the means are pairwise separated by at least $C \sqrt{\log k}$, then the means can be recovered to arbitrary accuracy (given enough samples).
They give an efficient algorithm which, warm-started with sufficiently-good estimates of the means, improves the accuracy to $\delta$ using $\poly(1/\delta,d,k)$ additional samples.
However, their algorithm for providing this warm start requires exponential time.
Our algorithm requires somewhat larger separation but runs in polynomial time.
Thus by combining the techniques in the spherical Gaussian setting we can estimate the means with $\ell_2$ error $\delta$ in polynomial time using an extra $\poly(1/\delta,d,k)$ samples, when the separation is at least $k^{\gamma}$, for any $\gamma > 0$.

\emph{Fixed number of Gaussians in many dimensions: }
Other works address parameter estimation for mixtures of $k \ll d$ Gaussians (generally $k = O(1)$ and $d$ grows) under weak identifiability assumptions \cite{DBLP:conf/stoc/KalaiMV10,DBLP:conf/focs/BelkinS10,DBLP:conf/focs/MoitraV10,DBLP:conf/stoc/HardtP15}.
In these works the only assumptions are that the component Gaussians are statistically distinguishable; the goal is to recover their parameters of the underlying Gaussians.
It was shown in \cite{DBLP:conf/stoc/HardtP15} that algorithms in this setting provably require $\exp(k)$ samples and running time.
The question addressed in our paper is whether this lower bound is avoidable under stronger identifiability assumptions.
A related line of work addresses proper learning of mixtures of Gaussians \cite{DBLP:conf/colt/FeldmanSO06, DK14, SOAJ14, LS17}, where the goal is to output a mixture of Gaussians which is close to the unknown mixture in total-variation distance, avoiding the $\exp(k)$ parameter-learning sample-complexity lower bound.
These algorithms achieve $\poly(k,d)$ sample complexity, but they all require $\exp(k)$ running time, and moreover, do not provide any guarantee that the parameters of the distributions output are close to those for the true mixture.

\emph{Tensor-decomposition methods: }
Another line of algorithms focus on settings where the means satisfy algebraic non-degeneracy conditions, which is the case for instance in smoothed analysis settings \cite{ MR3385380-Hsu13, DBLP:conf/colt/AndersonBGRV14, MR3388256-Ge15}.
These algorithms are typically based on finding a rank-one decomposition of the empirical $3$rd or $4$th moment tensor of the mixture; they heavily use the special structure of these moments for Gaussian mixtures.
One paper we highlight is \cite{DBLP:conf/stoc/BhaskaraCMV14}, which also uses much higher moments of the distribution.
They show that in the smoothed analysis setting, the $\ell$th moment tensor of the distribution has algebraic structure which can be algorithmically exploited to recover the means.
Their main structural result holds only in the smoothed analysis setting, where samples from a mixture model on perturbed means are available.

In contrast, we do not assume any non-degeneracy conditions and use moment information only about the individual components rather than the full mixture, which always hold under separation conditions.
Moreover, our algorithms do not need to know the exact structure of the 3rd or 4th moments.
In general, clustering-based algorithms like ours seem more robust to modelling errors than algebraic or tensor-decomposition methods.

\emph{Expectation-maximization (EM): }
EM is the most popular algorithm for Gaussian mixtures in practice, but it is notoriously difficult to analyze theoretically.
The works \cite{dasgupta2007probabilistic,daskalakis2016ten,xu2016global} offer some theoretical guarantees for EM, but non-convergence results are a barrier to strong theoretical guarantees \cite{wu1983convergence}.

\paragraph{Robust statistics}
The literature on robust estimation is too large to do justice to here.
There has been a long line of work on making algorithms tolerant to error in supervised settings \cite{DBLP:conf/ijcai/Valiant85,KL93}, especially for learning halfspaces \cite{Ser03,KLS09,ABL14,DKS17b}, and for problems such as PCA \cite{Bru09,CLMW11,LMTZ12,ZL14}.
See \cite{DBLP:conf/focs/DiakonikolasKK016} for a more detailed discussion on the relationship between these questions (and others) and the model we consider here.

We consider the classical statistical notion of robustness against corruption, introduced back in the 70's in seminal works of \cite{huber1964robust,Tuk75, HRRS86}.
Even for the mean of a Gaussian distribution, essentially all classical robust estimators are hard in the worst case to compute (\cite{JP78,Ber06}).
However, a recent flurry of work (\cite{DBLP:conf/focs/DiakonikolasKK016,DBLP:conf/focs/LaiRV16,DBLP:journals/corr/CharikarSV16,DKS16c,SCV17}) has given new, computationally efficient, nearly optimal robust estimators for the mean and covariance of a high dimensional Gaussian distribution.
Given sufficiently-many samples from a sub-Gaussian distribution with identity covariance, where an $\eps$-fraction are arbitrarily corrupted, these algorithms can output mean estimates which achieve error at most $O(\eps \sqrt{\log 1 / \eps})$ in $\ell_2$, which is information-theoretically optimal up to the $\sqrt{\log 1 / \eps}$ factor.
However, these mean estimation algorithms heavily rely on knowing that the covariance is equal (or very close) to the identity.
When the distribution is a general sub-Gaussian distribution with unknown covariance, the best known error achieved by an efficient algorithm is $O(\eps^{1/2})$ \cite{SCV17, DKKLMS17b}.
Under a slightly stronger assumption, our algorithm is able to achieve $O(\eps^{1 - 1/t})$ error in polynomial time, for arbitrarily large $t \in \N$, and error $O(\e \sqrt{\log 1/\e})$ in quasipolynomial time for distributions with $O(\log 1/\e)$ bounded moments.

\paragraph{SoS algorithms for unsupervised learning}
SoS algorithms for unsupervised learning obtain the best known polynomial-time guarantees for many problems, including dictionary learning, tensor completion, and others \cite{DBLP:conf/stoc/BarakKS14, DBLP:conf/stoc/BarakKS15, ge2015decomposing, DBLP:conf/colt/HopkinsSS15,DBLP:conf/focs/MaSS16,DBLP:conf/colt/BarakM16,DBLP:journals/corr/PotechinS17}.
While the running times of such algorithms are often large polynomials, due to the need to solve large SDPs, insights from the SoS algorithms have often been used in later works obtaining fast polynomial running times \cite{DBLP:conf/stoc/HopkinsSSS16, schramm-steurer,soslb}.
This lends hope that in light of our results there is a practical algorithm to learn mixture models under separation $k^{1/4 - \e}$ for some $\e > 0$.

\paragraph{Concurrent work}
Finally, we note that concurrent and independent works by several groups \cite{kothari-steinhardt-personal, kothari-steurer-personal, diakonikolas-kane-stewart-personal} have either obtained results or developed techniques similar to ours.

\subsection{Organization}
In Section~\ref{sec:techniques} we discuss at a high level the ideas in our algorithms and SoS proofs.
In Section~\ref{sec:prelims} we give standard background on SoS proofs.
Section~\ref{sec:overview-A} discusses the important properties of the family of polynomial inequalities we use in both algorithms.
Section~\ref{sec:mixture} and Section~\ref{sec:robust} state our algorithms formally and analyze them.
Finally, Section~\ref{sec:moment-polys} describes the polynomial inequalities our algorithms employ in more detail.

\section{Techniques}
\label{sec:techniques}
In this section we give a high-level overview of the main ideas in our algorithms.
First, we describe the proofs-to-algorithms methodology developed in recent work on SoS algorithms for unsupervised learning problems.
Then we describe the core of our algorithms for mixture models and robust estimation: a simple proof of identifiability of the mean of a distribution $\cD$ on $\R^d$ from samples $X_1,\ldots,X_n$ when some fraction of the samples may not be from $\cD$ at all.

\subsection{Proofs to algorithms for machine learning: the SoS method}
The Sum of Squares (SoS) hierarchy is a powerful tool in optimization, originally designed to approximately solve systems of polynomial equations via a hierarchy of increasingly strong but increasingly large semidefinite programming (SDP) relaxations (see \cite{DBLP:journals/corr/BarakS14} and the references therein).
There has been much recent interest in using the SoS method to solve unsupervised learning problems in generative models \cite{DBLP:conf/stoc/BarakKS14, DBLP:conf/stoc/BarakKS15, ge2015decomposing, DBLP:conf/colt/HopkinsSS15,DBLP:conf/focs/MaSS16,DBLP:journals/corr/PotechinS17}.
.

By now there is an established method for desgining such SoS-based algorithms, which we employ in this paper.
Consider a generic statistical estimation setting: there is a vector $\theta^* \in \R^k$ of parameters, and given some samples $x_1,\ldots,x_n \in \R^d$ sampled iid according to $p(x \, | \, \theta^*)$, one wants to recover some $\hat{\theta}(x_1,\ldots,x_n)$ such that $\|\theta^* - \hat{\theta} \| \leq \delta$ (for some appropriate norm $\|\cdot \|$ and $\delta \geq 0$).
One says that $\theta^*$ is \emph{identifiable} from $x_1,\ldots,x_n$ if, for any $\theta$ with $\|\theta^* - \theta\| > \delta$, one has $\Pr(x_1,\ldots,x_n \, | \, \theta') \ll \Pr(x_1,\ldots,x_n \, | \, \theta^*)$.
Often mathematical arguments for identifiability proceed via concentration of measure arguments culminating in a union bound over every possible $\theta$ with $\|\theta^* - \theta\| > \delta$.
Though this would imply $\theta$ could be recovered via brute-force search, this type of argument generally has no implications for efficient algorithms.

The SoS proofs-to-algorithms method prescribes designing a simple proof of identifiability of $\theta$ from samples $x_1,\ldots,x_n$.
Here ``simple'' has a formal meaning: the proof should be captured by the low-degree SoS proof system.
The SoS proof system can reason about equations and inequalities among low-degree polynomials.
Briefly, if $p(y_1,\ldots,y_m)$ and $q(y_1,\ldots,y_m)$ are polynomials with real coefficients, and for every $y \in \R^m$ with $p(y) \geq 0$ it holds also that $q(y) \geq 0$, the SoS proof system can deduce that $p(y) \geq 0$ implies $q(y) \geq 0$ if there is a simple certificate of this implication: polynomials $r(y), s(y)$ which are sums-of-squares, such that $q(y) = r(y) \cdot q(y) + s(y)$.
(Then $r,s$ form an SoS proof that $p(y) \geq 0$ implies $q(y) \geq 0$.)

Remarkably, many useful polynomial inequalities have such certificates.
For example, the usual proof of the Cauchy-Schwarz inequality $\iprod{y,z}^2 \leq \|y\|^2 \|z\|^2$, where $y,z$ are $m$-dimensional vectors, actually shows that the polynomial $\|y\|^2 \|z\|^2 - \iprod{y,z}^2$ is a sum-of-squares in $y$ and $z$.
The simplicity of the certificate is measured by the degree of the polynomials $r$ and $s$; when these polynomials have small (usually constant) degree there is hope of transforming SoS proofs into polynomial-time algorithms.
This transformation is possible because (under mild assumptions on $p$ and $q$) the set of low-degree SoS proofs is in fact captured by a polynomial-size semidefinite program.

Returning to unsupervised learning, the concentration/union-bound style of identifiability proofs described above are almost never captured by low-degree SoS proofs.
Instead, the goal is to design
\begin{enumerate}
\item A system of constant-degree polynomial equations and inequalties $\cA = \{ p_1(\theta)=0,\ldots,p_m(\theta)=0, q_1(\theta) \geq 0 ,\ldots,q_m(\theta) \geq 0 \}$, where the polynomials $p$ and $q$ depend on the samples $x_1,\ldots,x_n$, such that with high probability $\theta^*$ satisfies all the equations and inequalities.
\item A low-degree SoS proof that $\cA$ implies $\|\theta - \theta^*\| \leq \delta$ for some small $\delta$ and appropriate norm $\| \cdot \|$. \label{itm:sos-pf-overview}
\end{enumerate}
Clearly these imply that any solution $\theta$ of $\cA$ also solves the unsupervised learning problem.
It is in general NP-hard to find a solution to a system of low-degree polynomial equations and inequalities.

However, the SoS proof (\ref{itm:sos-pf-overview}) means that such a search can be avoided.
Instead, we will relax the set of solutions $\theta$ to $\cA$ to a simple(er) convex set: the set of \emph{pseudodistributions satisfying $\cA$}.
We define pseudodistributions formally later, for now saying only that they are the convex duals of SoS proofs which use the axioms $\cA$.
By this duality, the SoS proof (\ref{itm:sos-pf-overview}) implies not only that any solution $\theta$ to $\cA$ is a good choice of parameters but also that a good choice of parameters can be extracted any pseudodistribution satisfying $\cA$.
(We are glossing over for now that this last step requires some SDP rounding algorithm, since we use only standard rounding algorithms in this paper.)

Thus, the final SoS algorithms from this method take the form: solve an SDP to find a pseudodistribution which satisfies $\cA$ and round it to obtain a estimate $\hat{\theta}$ of $\theta^*$.
To analyze the algorithm, use the SoS proof (\ref{itm:sos-pf-overview}) to prove that $\| \hat{\theta} - \theta^* \| \leq \delta$.

\subsection{H\"older's inequality and identifiability from higher moments}
Now we discuss the core ideas in our simple SoS identifiability proofs.
We have not yet formally defined SoS proofs, so our goal will just be to construct identifiability proofs which are (a) phrased in terms of inequalities of low-degree polynomials and (b) provable using only simple inequalities, like Cauchy-Schwarz and H\"older's, leaving the formalities for later.

We consider an idealized version of situations we encounter in both the mixture model and robust estimation settings.
Let $\mu^* \in \R^d$.
Let $X_1,\ldots,X_n \in \R^d$ have the guarantee that for some $T \subseteq [n]$ of size $|T| = \alpha n$, the vectors $\{ X_i\}_{i \in T}$ are iid samples from $\cN(\mu^*, \Id)$, a spherical Gaussian centered at $\mu^*$; for the other vectors we make no assumption.
The goal is to estimate the mean $\mu^*$.

The system $\cA$ of polynomial equations and inequalities we employ will be designed so that a solution to $\cA$ corresponds to a subset of samples $S \subseteq [n]$ of size $|S| = |T| = \alpha n$.
We accomplish this by identifying $S$ with its $0/1$ indicator vector in $\R^n$ (this is standard).
The inequalities in $\cA$ will enforce the following crucial moment property on solutions: if $\mu = \tfrac 1 {|S|} \sum_{i \in S} X_i$ is the empirical mean of samples in $S$ and $t \in \N$, then
\begin{align}
  \frac 1 {|S|} \sum_{i \in S} \iprod{X_i - \mu, u}^t \leq 2 \cdot t^{t/2} \cdot \|u\|^t \qquad \text{ for all $u \in \R^d$}\mper \label{eq:moment-techniques}
\end{align}
This inequality says that every one-dimensional projection $u$ of the samples in $S$, centered around their empirical mean, has a sub-Gaussian empirical $t$-th moment.
(The factor $2$ accounts for deviations in the $t$-th moments of the samples.)
By standard concentration of measure, if $\alpha n \gg d^t$ the inequality holds for $S = T$.
It turns out that this property can be enforced by polynomials of degree $t$.
(Actually our final construction of $\cA$ will need to use inequalities of matrix-valued polynomials but this can be safely ignored here.)

Intuitively, we would like to show that any $S$ which satisfies $\cA$ has empirical mean close to $\mu^*$ using a low-degree SoS proof,.
This is in fact true when $\alpha = 1 - \e$ for small $\e$, which is at the core of our robust estimation algorithm.
However, in the mixture model setting, when $\alpha = 1/(\text{\# of components})$, for each component $j$ there is a subset $T_j \subseteq [n]$ of samples from component $j$ which provides a valid solution $S = T_j$ to $\cA$.
The empirical mean of $T_j$ is close to $\mu_j$ and hence not close to $\mu_i$ for any $i \neq j$.

We will prove something slightly weaker, which still demonstrates the main idea in our identifiability proof.
\begin{lemma}
\label{lem:overview-example}
With high probability, for every $S \subseteq [n]$, if $\mu = \tfrac 1 {|S|} \sum_{i \in S} X_i$ is the empirical mean of samples in $S$, then $\|\mu - \mu^*\| \leq 4t^{1/2} \cdot (|T|/|S \cap T|)^{1/t}$.
\end{lemma}
Notice that a random $S \subseteq [n]$ of size $\alpha n$ will have $|S \cap T| \approx \alpha^2 n$.
In this case the lemma would yield the bound $\|\mu - \mu^*\| \leq \tfrac {4t^{1/2}}{\alpha^{1/t}}$.
Thinking of $\alpha \ll 1/t$, this bound improves exponentially as $t$ grows.
In the $d$-dimensional $k$-component mixture model setting, one has $1/\alpha = \poly(k)$, and thus the bound becomes $\|\mu - \mu^*\| \leq 4t^{1/2} \cdot k^{O(1/t)}$.
In a mixture model where components are separated by $k^\e$, such an estimate is nontrivial when $\|\mu - \mu^*\| \ll k^\e$, which requires $t = O(1/\e)$.
This is the origin of the quantitative bounds in our mixture model algorithm.

We turn to the proof of Lemma~\ref{lem:overview-example}.
As we have already emphasized, the crucial point is that this proof will be accomplished using only simple inequalities, avoiding any union bound over all possible subsets $S$.

\begin{proof}[Proof of Lemma~\ref{lem:overview-example}]
Let $w_i$ be the $0/1$ indicator of $i \in S$.
To start the argument, we expand in terms of samples:
\begin{align}
|S \cap T| \cdot \|\mu - \mu^*\|^2 & = \sum_{i \in T} w_i \|\mu - \mu^*\|^2 \nonumber \\
& = \sum_{i \in T} w_i \iprod{\mu^* - \mu, \mu^* - \mu} \\
& = \sum_{i \in T} w_i \Brac{\iprod{X_i - \mu, \mu^* - \mu} + \iprod{\mu^* - X_i, \mu^* - \mu}}\mper
\end{align}
The key term to bound is the first one; the second amounts to a deviation term.
By H\"older's inequality and for even $t$,
\begin{align*}
  \sum_{i \in T} w_i \iprod{X_i - \mu, \mu^* - \mu} & \leq \Paren{\sum_{i \in T} w_i}^{\tfrac{t-1}t} \cdot \Paren{\sum_{i \in T}w_i \iprod{X_i - \mu, \mu^* - \mu}^t}^{1/t}\\
  & \leq \Paren{\sum_{i \in T} w_i}^{\tfrac{t-1}t} \cdot \Paren{\sum_{i \in [n]}w_i \iprod{X_i - \mu, \mu^* - \mu}^t}^{1/t}\\
  & \leq \Paren{\sum_{i \in T} w_i}^{\tfrac{t-1}t} \cdot 2t^{1/2} \cdot \|\mu^* - \mu\|\\
  & = |S \cap T|^{\tfrac{t-1} t}\cdot 2t^{1/2} \cdot \|\mu^* - \mu\|\mper
\end{align*}
The second line follows by adding the samples from $[n] \setminus T$ to the sum; since $t$ is even this only increases its value.
The third line uses the moment inequality \eqref{eq:moment-techniques}.
The last line just uses the definition of $w$.

For the second, deviation term, we use H\"older's inequality again:
\begin{align*}
  \sum_{i \in T} w_i \iprod{\mu^* - X_i, \mu^* - \mu} & \leq \Paren{\sum_{i \in T} w_i}^{\tfrac{t-1}t} \cdot \Paren{\sum_{i \in T} \iprod{\mu^* - X_i, \mu^* - \mu}^t}^{1/t}\mper
\end{align*}
The distribution of $\mu^* - X_i$ for $i \in T$ is $\cN(0,\Id)$.
By standard matrix concentration, if $|T| = \alpha n \gg d^t$,
\[
\sum_{i \in T} \Brac{(X_i - \mu^*)^{\tensor t/2}} \Brac{(X_i - \mu^*)^{\tensor t/2}}^{\top} \preceq 2 |T| \E_{Y \sim \cN(0,\Id)} \Paren{Y^{\tensor t/2}} \Paren{Y^{\tensor t/2}}^{\top}
\]
with high probability and hence, using the quadratic form at $(\mu^* - \mu)^{\tensor t/2}$,
\[
  \sum_{i \in T} \iprod{\mu^* - X_i, \mu^* - \mu}^t \leq 2 |T| t^{t / 2} \cdot \|\mu^* - \mu\|^t\mper
\]

Putting these together and simplifying constants, we have obtained that with high probability,
\[
  |S \cap T| \cdot \|\mu - \mu^*\|^2 \leq 4t^{t/2} |T|^{1/t} \cdot |S \cap T|^{(t-1)/t} \cdot \|\mu - \mu^*\|
\]
which simplifies to
\[
 |S \cap T|^{1/t} \cdot \|\mu - \mu^*\| \leq 4t^{1/2} |T|^{1/t}\mper\qedhere
\]
\end{proof}

\subsection{From identifiability to algorithms}
We now discuss how to use the ideas described above algorithmically for learning well-separated mixture models.
The high level idea for robust estimation is similar.
Given Lemma~\ref{lem:overview-example}, a naive algorithm for learning mixture models would be the following: find a set of points $T$ of size roughly $n / k$ that satisfy the moment bounds described, and simply output their empirical mean.
Since by a simple counting argument this set must have nontrivial overlap with the points from some mixture component, Lemma~\ref{lem:overview-example} guarantees that the empirical mean is close to mean of this component.

However, in general finding such a set of points is algorithmically difficult.
In fact, it would suffice to find a distribution over such sets of points (since then one could simply sample from this distribution), however, this is just as computationally difficult.
The critical insight is that because of the proof of Lemma~\ref{lem:overview-example} only uses facts about low degree polynomials, it suffices to find an object which is indistinguishable from such a distribution, considered as a functional on low-degree polynomials.

The natural object in this setting is a \emph{pseudo-distribution}.
Pseudo-distributions form a convex set, and for a set of low-degree polynomial equations and inequalities $\cA$, it is possible to find a pseudo-distribution which is indistinguishable from a distribution over solutions to $\cA$ (as such a functional) in polynomial time via semidefinite programming (under mild assumptions on $\cA$).
More specifically, the set of SoS proofs using axioms $\cA$ is a semidefinite program (SDP), and the above pseudodistributions form the dual SDP.
(We will make these ideas more precise in the next two sections.)

Our algorithm then proceeds via the following general framework: find an appropriate pseudodistribution via convex optimization, then leverage our low-degree sum of squares proofs to show that information about the true clusters can be extracted from this object by a standard SDP rounding procedure.


\section{Preliminaries}
\label{sec:prelims}

Throughout the paper we let $d$ be the dimensionality of the data, and we will be interested in the regime where $d$ is at least a large constant.
We also let $\| v \|$ denote the $\ell_2$ norm of a vector $v$, and $\| M \|_F$ to denote the Frobenius norm of a matrix $M$; often we just write $\|M\|$.
We will also give randomized algorithms for our problems that succeed with probability $1 - \poly(1 / k,1/d)$; by standard techniques this probability can be boosted to $1 - \xi$ by increasing the sample and runtime complexity by a mulitplicative $\log 1 / \xi$.

We now formally define the class of distributions we will consider throughout this paper.
At a high level, we will consider distributions which have bounded moments, for which there exists a low degree SoS proof of this moment bound.
Formally:
\begin{definition}
\label{def:explicitly-bounded}
Let $\cD$ be a distribution over $\R^d$ with mean $\mu$.
For $c \geq 1, t \in \N$, we say that $\cD$ is $t$-explicitly bounded with variance proxy $\sigma$ if for every even $s \leq t$ there is a degree $s$ SoS proof (see Section~\ref{sec:sos} for a formal definition) of
\[
  \proves{s} E_{Y \sim \cD_k} \iprod{\Paren{Y - \mu},u}^s \leq (\sigma s)^{s / 2} \|u\|^s \mper
\]
Equivalently, the polynomial $p(u) = (\sigma s)^{s / 2} \|u\|^s - E_{Y \sim \cD_k} \iprod{\Paren{Y - \mu},u}^s $ should be a sum-of-squares.
In our typical use case, $\sigma = 1$, we will omit it and call the distribution $t$-explicitly bounded.
\end{definition}
\noindent
Throughout this paper, since all of our problems are scale invariant, we will assume without loss of generality that $\sigma = 1$.
This class of distributions captures a number of natural classes of distributions.
Intuitively, if $u$ were truly a vector in $\R^k$ (rather than a vector of indeterminants), then this exactly captures sub-Gaussian type moment.
Our requirement is simply that these types of moment bounds not only hold, but also have a SoS proof.

We remark that our results also hold for somewhat more general settings.
It is not particularly important that the $s$-th moment bound has a degree $s$ proof; our techniques can tolerate degree $O(s)$ proofs.
Our techniques also generally apply for weaker moment bounds.
For instance, our techniques naturally extend to explicitly bounded sub-exponential type distributions in the obvious way.
We omit these details for simplicity.

As we show in Appendix~\ref{sec:explicitly-bounded-families}, this class still captures many interesting types of nice distributions, including Gaussians, product distributions with sub-Gaussian components, and rotations therof.
With this definition in mind, we can now formally state the problems we consider in this paper:

\paragraph{Learning well-separated mixture models}
We first define the class of mixture models for which our algorithm works:
\begin{definition}[$t$-explicitly bounded mixture model with separation $\Delta$]
  Let $\mu_1,\ldots,\mu_k \in \R^d$ satisfy $\|\mu_i - \mu_j\| > \Delta$ for every $i \neq j$, and let $\cD_1,\ldots,\cD_k$ have means $\mu_1,\ldots,\mu_k$, so that each $\cD_i$ is $t$-explicitly bounded.
  Let $\lambda_1,\ldots,\lambda_k \geq 0$ satisfy $\sum_{i \in [k]} \lambda_i = 1$.
  Together these define a mixture distribution on $\R^d$ by first sampling $i \sim \lambda$, then sampling $x \sim \cD_i$.
\end{definition}
\noindent{}
The problem is then:
\begin{problem}
Let $\cD$ be a $t$-explicitly bounded mixture model in $\R^d$ with separation $\Delta$ with $k$ components.
Given $k, \Delta$, and $n$ independent samples from $\cD$, output $\muhat_1, \ldots, \muhat_m$ so that with probability at least $0.99$, there exists a permutation $\pi: [k] \to [k]$ so that $\| \mu_i - \muhat_{\pi(i)} \| \leq \delta$ for all $i = 1, \ldots, k$.
\end{problem}

\paragraph{Robust mean estimation}
We consider the same basic model of corruption introduced in \cite{DBLP:conf/focs/DiakonikolasKK016}.
\begin{definition}[$\eps$-corruption]
We say a set of samples $X_1, \ldots, X_n$ is $\eps$-corrupted from a distribution $\cD$ if they are generated via the following process.
First, $n$ independent samples are drawn from $\cD$.
Then, an adversary changes $\eps n$ of these points arbitrarily, and the altered set of points is then returned to us in an arbitrary order.
\end{definition}
\noindent
The problem we consider in this setting is the following:
\begin{problem}[Robust mean estimation]
Let $\cD$ be an $O(t)$-explicitly bounded distribution over $\R^d$ wih mean $\mu$.
Given $t, \eps$, and an $\eps$-corrupted set of samples from $\cD$, output $\muhat$ satisfying $\| \mu - \muhat \| \leq O(\eps^{1 - 1 / t})$.
\end{problem}

\subsection{The SoS proof system}
\label{sec:sos}
We refer the reader to \cite{DBLP:conf/soda/ODonnellZ13,DBLP:journals/corr/BarakS14} and the references therein for a thorough exposition of the SoS algorithm and proof system; here we only define what we need.\footnote{Our definition of SoS proofs differs slightly from O'Donnell and Zhou's in that we allow proofs to use products of axioms.}

Let $x_1,\ldots,x_n$ be indeterminates and $\cA$ be the set of polynomial equations and inequalities $\{p_1(x) \geq 0,\ldots,p_m(x) \geq 0, q_1(x) = 0,\ldots,q_m(x) = 0\}$.
We say that the statement $p(x) \geq 0$ has an SoS proof if there are polynomials $\{r_\alpha\}_{\alpha \subseteq [m]}$ (where $\alpha$ may be a multiset) and $\{s_i\}_{i \in [m]}$ such that
\[
  p(x) = \sum_{\alpha} r_\alpha(x) \cdot \prod_{i \in \alpha} p_i(x) + \sum_{i \in [m]} s_i(x) q_i(x)
\]
and each polynomial $r_\alpha(x)$ is a sum of squares.

If the polynomials $r_\alpha(x) \cdot \prod_{i \in \alpha} p_i(x)$ and $s_i(x) q_i(x)$ have degree at most $d$, we say the proof has degree at most $d$, and we write
\[
  \cA \proves{d} p(x) \geq 0\mper
\]

SoS proofs compose well, and we frequently use the following without comment.
\begin{fact}
  If $\cA \proves{d} p(x) \geq 0$ and $\cA \proves{d'} q(x) \geq 0$, then $\cA \cup \calB \proves{\max(d,d')} p(x) + q(x) \geq 0$ and $\cA \cup \calB \proves{dd'} p(x) q(x) \geq 0$.
\end{fact}

We turn to the dual objects to SoS proofs.
A degree-$d$ pseudoexpectation (for variety we sometimes say ``pseudodistribution") is a linear operator $\pE \, : \, \R[x]_{\leq d} \rightarrow \R$, where $\R[x]_{\leq d}$ are the polynomials in indeterminates $x$ with real coefficients, which satisfies the following
\begin{enumerate}
  \item Normalization: $\pE[1] = 1$
  \item Positivity: $\pE[p(x)^2] \geq 0$ for every $p$ of degree at most $d/2$.
\end{enumerate}
We say that a degree-$d$ pseudoexpectation $\pE$ satisfies inequalities and equalities $\{p_1(x) \geq 0,\ldots,p_m(x) \geq 0, q_1(x) = 0,\ldots,q_m(x) = 0\}$ if
\begin{enumerate}
  \item for every multiset $\alpha \subseteq [m]$ and SoS polynomial $s(x)$ such that the degree of $s(x) \prod_{i \in \alpha} p_i(x)$ is at most $d$, one has $\pE s(x) \prod_{i \in \alpha} p_i(x) \geq 0$, and
  \item for every $q_i(x)$ and every polynomial $s(x)$ such that the degree of $q_i(x) s(x) \leq d$, one has $\pE s(x) q_i(x) = 0$.
\end{enumerate}

The main fact relating pseudoexpectations and SoS proofs is:
\begin{fact}[Soundness of SoS proofs]
  If $\cA$ is a set of equations and inequalities and $\cA \proves{d} p(x) \geq 0$, and $\pE$ satisfies $\cA$, then $\pE$ satisfies $\cA \cup \{ p \geq 0\}$.
\end{fact}
\noindent
In Section~\ref{sec:toolkit} we state and prove many basic SoS inequalities that we will require throughout the paper.

\paragraph{Gaussian distributions are explicitly bounded}
In Section~\ref{sec:toolkit} we show that product distributions (and rotations thereof) with bounded $t$-th moments are explicitly bounded.
\begin{lemma}
Let $\cD$ be a distribution over $\R^d$ so that $\cD$ is a rotation of a product distribution $\cD'$ where each coordinate $X$ with mean $\mu$ of $\cD$ satisfies
\[
  \E [(X - \mu)^{s}] \leq 2^{-s} \Paren{ \frac s2}^{s / 2}
\]
Then $\cD$ is $t$-explicitly bounded (with variance proxy $1$).
\end{lemma}
\noindent (The factors of $\frac 12$ can be removed for many distributions, including Gaussians.)


\section{Capturing empirical moments with polynomials}
\label{sec:overview-A}
To describe our algorithms we need to describe a system of polynomial equations and inequalities which capture the following problem: among $X_1,\ldots,X_n \in \R^d$, find a subset of $S \subseteq [n]$ of size $\alpha n$ such that the empirical $t$-th moments obey a moment bound: $\tfrac 1 {\alpha n} \sum_{i \in S} \iprod{X_i,u}^t \leq t^{t/2} \|u\|^t$ for every $u \in \R^d$.

Let $k, n \in \N$ and let $w = (w_1,\ldots,w_n), \mu = (\mu_1,\ldots,\mu_k)$ be indeterminates.
Let
\begin{compactenum}
  \item $X_1,\ldots,X_n \in \R^d$
  \item $\alpha \in [0,1]$ be a number (the intention is $|S| = \alpha n$).
  \item $t \in \N$ be a power of $2$, the order of moments to control
  \item $\mu_1,\ldots,\mu_k \in \R^d$, which will eventually be the means of a $k$-component mixture model, or when $k=1$, the true mean of the distribution whose mean we robustly estimate.
  \item $\tau > 0$ be some error magnitude accounting for fluctuations in the sizes of clusters (which may be safely ignored at first reading).
\end{compactenum}

\begin{definition}
Let $\cA$ be the following system of equations and inequalities, depending on all the parameters above.
\begin{enumerate}
  \item $w_i^2 = w_i$ for all $i \in [n]$ (enforcing that $w$ is a $0/1$ vector, which we interpret as the indicator vector of the set $S$).
  \item $(1 - \tau) \alpha n \leq \sum_{i \in [n]} w_i \leq (1 + \tau) \alpha n$, enforcing that $|S| \approx \alpha n$ (we will always choose $\tau = o(1)$).
  \item $\mu \cdot \sum_{i \in [n]} w_i = \sum_{i \in [n]} w_i X_i$, enforcing that $\mu$ is the empirical mean  of the samples in $S$
  \item $\sum_{i \in [n]} w_i \iprod{X_i - \mu, \mu - \mu_j}^t \leq 2 \cdot t^{t / 2} \sum_{i \in [n]} w_i \|\mu - \mu_j\|^t$ for every $\mu_j$ among $\mu_1,\ldots,\mu_m$.
  This enforces that the $t$-th empirical moment of the samples in $S$ is bounded \emph{in the direction $\mu - \mu_j$}.
  \end{enumerate}
\end{definition}
Notice that since we will eventually take $\mu_j$'s to be unknown parameters we are trying to estimate, the algorithm cannot make use of $\cA$ directly, since the last family of inequalities involve the $\mu_j$'s.
  Later in this paper we exhibit a system of inequalities which requires the empirical $t$-th moments to obey a sub-Gaussian type bound in every direction, hence implying the inequalities here without requiring knowledge of the $\mu_j$'s to write down.
  Formally, we will show:
  \begin{lemma}
  \label{lem:general-structured-subset-polys}
  Let $\alpha \in [0,1]$.
  Let $t \in \N$ be a power of $2$, $t \geq 4$.\footnote{The condition $t \geq 4$ is merely for technical convenience.}
  Let $0.1 > \tau > 0$.
  Let $X_1,\ldots,X_n \in \R^d$.
  Let $\cD$ be a $10t$-explicitly bounded distribution.

  There is a family $\cAhat$ of polynomial equations and inequalities of degree $O(t)$ on variables $w = (w_1,\ldots,w_n), \mu = (\mu_1,\ldots,\mu_k)$ and at most $n^{O(t)}$ other variables, whose coefficients depend on $\alpha,t,\tau,X_1,\ldots,X_n$, such that
  \begin{enumerate}
    \item (Satisfiability) If there $S \subseteq [n]$ of size at least $(\alpha - \tau)n$ so that $\{X_i\}_{i \in S}$ is an iid set of samples from $\cD$, and $(1 - \tau) \alpha n \geq d^{100 t}$, then for $d$ large enough, with probability at least $1 - d^{-8}$, the system $\cAhat$ has a solution over $\R$ which takes $w$ to be the $0/1$ indicator vector of $S$. \label{enum:satisfiable}
  \item (Solvability) For every $C \in \N$ there is an $n^{O(Ct)}$-time algorithm which, when $\cAhat$ is satisfiable, returns a degree-$Ct$ pseudodistribution which satisfies $\cAhat$ (up to additive error $2^{-n}$). \label{itm:solve}
    \item (Moment bounds for polynomials of $\mu$) Let $f(\mu)$ be a length-$d$ vector of degree-$\ell$ polynomials in indeterminates $\mu = (\mu_1,\ldots,\mu_k)$.
  $\cAhat$ implies the following inequality and the implication has a degree $t\ell$ SoS proof.
  \begin{align*}
    \cAhat \proves{O(t \ell)} \frac 1 {\alpha n} \sum_{i \in [n]} w_i \iprod{X_i - \mu, f(\mu)}^t \leq 2 \cdot t^{t / 2} \|f(\mu)\|^t \mper
  \end{align*}
  \label{enum:moment-bounds}
  \item (Booleanness) $\cAhat$ includes the equations $w_i^2 = w_i$ for all $i \in [n]$. \label{enum:boolean}
  \item (Size) $\cAhat$ includes the inequalities $(1 - \tau) \alpha n \leq \sum w_i \leq (1 + \tau) \alpha n$. \label{enum:size}
  \item (Empirical mean) $\cAhat$ includes the equation $\mu \cdot \sum_{i \in [n]} w_i = \sum_{i \in [n]} w_i X_i$. \label{enum:empirical-mean}
  \end{enumerate}
  In particular this implies that $\cAhat \proves{O(t)} \cA$.
\end{lemma}
The proof of Lemma~\ref{lem:general-structured-subset-polys} can be found in Section~\ref{sec:moment-polys}.

\begin{remark}[Numerical accuracy, semidefinite programming, and other monsters]
  We pause here to address issues of numerical accuracy.
  Our final algorithms use point \ref{itm:solve} in Lemma~\ref{lem:general-structured-subset-polys} (itself implemented using semidefinite programming) to obtain a pseudodistribution $\pE$ satisfying $\cAhat$ approximately, up to error $\eta = 2^{-n}$ in the following sense:
  for every $r$ a sum of squares and $f_1,\ldots,f_\ell \in \cA$ with $\deg \Brac{r \cdot \prod f_i \leq Ct}$, one has $\pE r \cdot \prod_{i \in \cA} f \geq -\eta \cdot \|r\|$, where $\|r\|$ is $\ell_2$ norm of the coefficients of $r$.
  Our main analyses of this pseudodistribution employ the implication $\cAhat \proves{} \calB$ for another family of inequalities $\calB$ to conclude that if $\pE$ satisfies $\cA$ then it satisfies $\calB$, then use the latter to analyze our rounding algorithms.
  Because all of the polynomials eventually involved in the SoS proof $\cAhat \proves{} \calB$ have coefficients bounded by $n^B$ for some large constant $B$, it may be inferred that if $\pE$ approximately satisfies $\cAhat$ in the sense above, it also approximately satisfies $\calB$, with some error $\eta' \leq 2^{-\Omega(n)}$.
  The latter is a sufficient for all of our rounding algorithms.

  Aside from mentioning at a couple key points why our SoS proofs have bounded coefficients, we henceforth ignore all numerical issues.
  For further discussion of numerical accuracy and well-conditioned-ness issues in SoS, see \cite{o2016sos,sos-notes-12,DBLP:journals/corr/RaghavendraW17}
\end{remark}


\section{Mixture models: algorithm and analysis}
\label{sec:mixture}

In this section we formally describe and analyze our algorithm for mixture models.
We prove the following theorem.

\begin{theorem}[Main theorem on mixture models]
  \label{thm:mixture-main}
  For every large-enough $t \in \N$ there is an algorithm with the following guarantees.
  Let $\mu_1,\ldots,\mu_k \in \R^d$, satisfy $\|\mu_i - \mu_j\| \geq \Delta$.
  Let $\cD_1,\ldots,\cD_k$ be $10 t$-explicitly bounded, with means $\mu_1,\ldots,\mu_k$.
  Let $\lambda_1,\ldots,\lambda_k \geq 0$ satisfy $\sum \lambda_i = 1$.
  Given $n \geq (d^tk)^{O(1)} \cdot (\max_{i \in [m]} 1/\lambda_i)^{O(1)}$ samples from the mixture model given by $\lambda_1,\ldots,\lambda_k, \cD_1,\ldots,\cD_k$, the algorithm runs in time $n^{O(t)}$ and with high probability returns $\{\hat{\mu}_1,\ldots,\hat{\mu}_k \}$ (not necessarily in that order) such that
  \[
    \|\mu_i - \hat{\mu}_i\| \leq \frac{2^{Ct} m^C t^{t / 2}}{\Delta^{t-1}}
  \]
  for some universal constant $C$.
\end{theorem}
  In particular, we note two regimes: if $\Delta = k^\gamma$ for a constant $\gamma > 0$, choosing $t = O(1 / \gamma)$ we get that the $\ell_2$ error of our estimator is $\poly(1/k)$ for any $O(1/ \gamma)$-explicitly bounded distribution, and our estimator requires only $(dk)^{O(1)}$ samples and time.
  This matches the guarantees of Theorem~\ref{thm:mixtures-intro}.

  On the other hand, if $\Delta = C' \sqrt{\log k}$ (for some universal $C'$) then taking $t = O(\log k)$ gives error
  \[
    \| \mu_i - \hat{\mu}_i \| \leq k^{O(1)} \cdot \Paren{\frac{\sqrt t}{\Delta}}^t
  \]
  which, for large-enough $C'$ and $t$, can be made $1/\poly(k)$.
  Thus for $\Delta = C' \sqrt{\log k}$ and any $O(\log k)$-explicitly bounded distrituion we obtain error $1/\poly(k)$ with $d^{O(\log k)}$ samples and $d^{O(\log k)^2}$ time.

  In this section we describe and analyze our algorithm.
To avoid some technical work we analyze the uniform mixtures setting, with $\lambda_i = 1/m$.
In Section~\ref{sec:nonuniform-mixtures} we describe how to adapt the algorithm to the nonuniform mixture setting.

\subsection{Algorithm and main analysis}
We formally describe our mixture model algorithm now.
We use the following lemma, which we prove in Section~\ref{sec:rounding}.
The lemma says that given a matrix which is very close, in Frobenious norm, to the $0/1$ indicator matrix of a partition of $[n]$ it is possible to approximately recover the partition.
(The proof is standard.)

\begin{lemma}[Second moment rounding, follows from Theorem~\ref{thm:round-main}]
  \label{lem:rounding-main}
  Let $n,m \in \N$ with $m \ll n$.
  There is a polynomial time algorithm \textsc{RoundSecondMoments} with the following guarantees.
  Suppose $S_1,\ldots,S_m$ partition $[n]$ into $m$ pieces, each of size $\tfrac n{2m} \leq |S_i| \leq  \tfrac {2n} m$.
  Let $A \in \R^{n \times n}$ be the $0/1$ indicator matrix for the partition $S$; that is, $A_{ij} = 1$ if $i,j \in S_\ell$ for some $\ell$ and is $0$ otherwise.
  Let $M \in \R^{n \times n}$ be a matrix with $\|A - M\|_F \leq \e n$.
  Given $M$, with probability at least $1 - \e^2 m^3$ the algorithm returns a partition $C_1,\ldots,C_m$ of $[n]$ such that up to a global permutation of $[m]$, $C_i = T_i \cup B_i$, where $T_i \subseteq S_i$ and $|T_i| \geq |S_i| - \e^2 m^2 n$ and $|B_i| \leq \e^2 m^2 n$.
\end{lemma}

\begin{algorithm}[htb]{}
\begin{algorithmic}[1]
\Function{LearnMixtureMeans}{$t,X_1,\ldots,X_n,\delta,\tau$}
\State By semidefinite programming (see Lemma~\ref{lem:general-structured-subset-polys}, item \ref{itm:solve}), find a pseudoexpectation of degree $O(t)$ which satisfies the structured subset polynomials from Lemma~\ref{lem:general-structured-subset-polys}, with $\alpha = n/m$ such that $\|\pE ww^\top \|_F$ is minimized among all such pseudoexpectations.
\State Let $M \gets m \cdot \pE ww^\top$.
\State Run the algorithm \textsc{RoundSecondMoments} on $M$ to obtain a partition $C_1,\ldots,C_m$ of $[n]$.
\State Run the algorithm \textsc{EstimateMean} from Section~\ref{sec:robust} on each cluster $C_i$, with $\e = 2^{Ct} t^{t/2} m^4 / \Delta^t$ for some universal constant $C$ to obtain a list of mean estimates $\hat{\mu}_1,\ldots,\hat{\mu}_m$.
\State Output $\hat{\mu}_1,\ldots,\hat{\mu}_m$.
\EndFunction
\end{algorithmic}
\caption{Mixture Model Learning}
\label{alg:learn-mixture}
\end{algorithm}

\begin{remark}[On the use of \textsc{EstimateMean}]
As described, \textsc{LearnMixtureMeans} has two phases: a clustering phase and a mean-estimation phase.
The clustering phase is the heart of the algorithm; we will show that after running \textsc{RoundSecondMoments} the algorithm has obtained clusters $C_1,\ldots,C_k$ which err from the ground-truth clustering on only a $\tfrac{2^{O(t)} t^{t/2} \poly(k)}{\Delta^t}$-fraction of points.
To obtain estimates $\hat{\mu}_i$ of the underlying means from such a clustering, one simple option is to output the empirical mean of the clusters.
However, without additional pruning this risks introducing error in the mean estimates which grows with the ambient dimension $d$.
By using the robust mean estimation algorithm instead to obtain mean estimates from the clusters we obtain errors in the mean estimates which depend only on the number of clusters $k$, the between-cluster separation $\Delta$, and the number $t$ of bounded moments.
\end{remark}

\begin{remark}[Running time]
We observe that \textsc{LearnMixtureMeans} can be implemented in time $n^{O(t)}$.
The main theorem requires $n \geq k^{O(1)} d^{O(t)}$, which means that the final running time of the algorithm is $(kd^t)^{O(t)}$.\footnote{As discussed in Section~\ref{sec:overview-A}, correctness of our algorithm at the level of numerical accuracy requires that the coefficients of every polynomial in the SoS program $\cAhat$ (and every polynomial in the SoS proofs we use to analyze $\cAhat$) are polynomially bounded.
This may not be the case if some vectors $\mu_1,\ldots,\mu_m$ have norms $\|\mu_i\| \geq d^{\omega(1)}$.
This can be fixed by naively clustering the samples $X_1,\ldots,X_n$ via single-linkage clustering, then running \textsc{LearnMixtureMeans} on each cluster.
It is routine to show that the diameter of each cluster output by a naive clustering algorithm is at most $\poly(d,k)$ under our assumptions, and that with high probability single-linkage clustering produces a clustering respecting the distributions $\cD_i$.
Hence, by centering each cluster before running \textsc{LearnMixtureMeans} we can assume that $\|\mu_i\| \leq \poly(d,k)$ for every $i \leq d$.}
\end{remark}

\subsection{Proof of main theorem}
In this section we prove our main theorem using the key lemmata; in the following sections we prove the lemmata.

\paragraph{Deterministic Conditions}
We recall the setup.
There are $k$ mean vectors $\mu_1,\ldots,\mu_k \in \R^d$, and corresponding distributions $\cD_1,\ldots,\cD_k$ where $\cD_j$ has mean $\mu_j$.
The distributions $\cD_j$ are $10t$-explicitly bounded for a choice of $t$ which is a power of $2$.
Vectors $X_1,\ldots,X_n \in \R^d$ are samples from a uniform mixture of $\cD_1,\ldots,\cD_k$.
We will prove that our algorithm succeeds under the following condition on the samples $X_1,\ldots,X_n$.

\begin{enumerate}
\item[(D1)] (Empirical moments) For every cluster $S_j = \{X_i \, : \, X_i \text{ is from $\cD_j$} \}$, the system $\cAhat$ from Lemma~\ref{lem:general-structured-subset-polys} with $\alpha = 1/m$ and $\tau = \Delta^{-t}$ has a solution which takes $w \in \{0,1\}^n$ to be the $0/1$ indicator vector of $S_j$.
\item[(D2)] (Empirical means) Let $\bmu_j$ be the empirical mean of cluster $S_j$.
The $\bmu_j$'s satisfy $\|\bmu_i - \mu_i \| \leq \Delta^{-t}$.
\end{enumerate}
We note a few useful consequences of these conditions, especially (D1).
First of all, it implies all clusters have almost the same size: $(1 - \Delta^{-t}) \cdot \tfrac nk \leq |S_j| \leq (1 + \Delta^{-t}) \cdot \tfrac nk$.
Second, it implies that all clusters have explicitly bounded moments: for every $S_j$,
\[
  \proves{t} \frac kn \sum_{i \in S_j} \iprod{X_i - \bmu_j, u}^t \leq 2 \cdot t^{t/2} \cdot \|u\|^t\mper
\]

\paragraph{Lemmas}

The following key lemma captures our SoS identifiability proof for mixture models.
\begin{lemma}\label{lem:mixture-ident-main}
  Let $\mu_1,\ldots,\mu_k,\cD_1,\ldots,\cD_k$ be as in Theorem~\ref{thm:mixture-main}, with mean separation $\Delta$.
  Suppose (D1), (D2) occur for samples $X_1,\ldots,X_n$.
  Let $t \in \N$ be a power of two.
  Let $\pE$ be a degree-$O(t)$ pseudoexpectation which satisfies $\cA$ from Lemma~\ref{lem:general-structured-subset-polys} with $\alpha = 1/k$ and $\tau \leq \Delta^{-t}$.
  Then for every $j,\ell \in [k]$,
  \[
    \pE \iprod{a_j,w} \iprod{a_\ell,w} \leq 2^{8t+8} \cdot t^{t/2} \cdot \frac{n^2}{k} \cdot \frac 1 {\Delta^t}\mper
   \]
\end{lemma}
\noindent
The other main lemma shows that conditions (D1) and (D2) occur with high probability.
\begin{lemma}[Concentration for mixture models]
\label{lem:mixture-concentration}
  With notation as above, conditions (D1) and (D2) simultaneously occur with probability at least $1 - 1/d^{15}$ over samples $X_1,\ldots,X_n$, so long as $n \geq d^{O(t)} k^{O(1)}$, for $\Delta \geq 1$.
\end{lemma}
\noindent Lemma~\ref{lem:mixture-concentration} follows from Lemma~\ref{lem:general-structured-subset-polys}, for (D1), and standard concentration arguments for (D2).
Now we can prove the main theorem.
\begin{proof}[Proof of Theorem~\ref{thm:mixture-main} (uniform mixtures case)]
  Suppose conditions (D1) and (D2) hold.
  Our goal will be to bound $\|M - A\|^2 \leq n \cdot \tfrac{2^{O(t)} t^{t /2} k^4}{\Delta^{t}}$, where $A$ is the $0/1$ indicator matrix for the ground truth partition $S_1,\ldots,S_k$ of $X_1,\ldots,X_n$ according to $\cD_1,\ldots,\cD_k$.
  Then by Lemma~\ref{lem:rounding-main}, the rounding algorithm will return a partition $C_1,\ldots,C_k$ of $[n]$ such that $C_\ell$ and $S_\ell$ differ by at most $n \tfrac{2^{O(t)} t^{t/2} k^{10}}{\Delta^{t}}$ points, with probability at least $1 - \tfrac{2^{O(t)} t^{t/2} k^{30}}{\Delta^{t}}$.
  By the guarantees of Theorem~\ref{thm:robust-main} regarding the algorithm \textsc{EstimateMean}, with high probability the resulting error in the mean estimates $\hat{\mu}_i$ will satisfy
  \[
  \|\mu_i - \hat{\mu}_i\| \leq \sqrt t \cdot \Paren{\frac{2^{O(t)} t^{t/2} k^{10}}{\Delta^t}}^{\tfrac {t-1}t} \leq \frac{2^{O(t)} \cdot t^{t/2} \cdot k^{10}}{\Delta^{t-1}}\mper
  \]

  We turn to the bound on $\|M - A\|^2$.
  First we bound $\iprod{\pE ww^\top, A}$.
  Getting started,
  \[
   \pE \Paren{\sum_{i \in [k]} \iprod{w,a_i}}^2 = \pE \Paren{ \sum_{i \in [n]} w_i}^2 \geq (1 - \Delta^{-t} )^2 \cdot n^2/k^2\mper
  \]
  By Lemma~\ref{lem:mixture-ident-main}, choosing $t$ later,
  \[
    \sum_{i \neq j \in [k]} \pE \iprod{a_i,w}\iprod{a_j,w} \leq n^2 2^{O(t)} t^{t/2} \cdot k \cdot \frac 1 {\Delta^{t}}\mper
  \]
  Together, these imply
  \[
    \pE \sum_{i \in [k]} \iprod{w,a_i}^2 \geq \frac{n^2}{k^2} \cdot \Brac{ 1- \frac{2^{O(t)} t^{t /2} k^3}{\Delta^{t}}}\mper
  \]

  At the same time, $\|\pE ww^T\|_F \leq \tfrac 1k \|A\|_F$ by minimality (since the uniform distribution over cluster indicators satisfies $\cA$), and by routine calculation and assumption (D1), $\|A\|_F \leq \tfrac n {\sqrt k} (1 + O(\Delta^{-t}))$.
  Together, we have obtained
  \[
  \iprod{M,A} \geq \Paren{1 - \frac{2^{O(t)} t^{t/2} k^3}{\Delta^{t}}} \cdot \|A\| \|M\|
  \]
  which can be rearranged to give $\|M-A\|^2 \leq n \cdot \tfrac{2^{O(t)} t^{t/2} k^4}{\Delta^{t}}$.
\end{proof}

\subsection{Identifiability}
In this section we prove Lemma~\ref{lem:mixture-ident-main}.
We use the following helpful lemmas.
The first is in spirit an SoS version of Lemma~\ref{lem:overview-example}.
\begin{lemma}\label{lem:mixture-ident-dist}
  Let $\mu_1,\ldots,\mu_k,\cD_1,\ldots,\cD_k,t$ be as in Theorem~\ref{thm:mixture-main}.
  Let $\bmu_i$ be as in (D1).
  Suppose (D1) occurs for samples $X_1,\ldots,X_n$.
  Let $\cA$ be the system from Lemma~\ref{lem:general-structured-subset-polys}, with $\alpha = 1/k$ and any $\tau$.
  Then
  \[
    \cA \proves{O(t)} \iprod{a_j, w}^t \|\mu - \bmu_j\|^{2t} \leq 2^{t+2} t^{t/2} \cdot \frac nk \cdot \iprod{a_j,w}^{t-1} \cdot \|\mu - \bmu_j\|^t\mper
  \]
\end{lemma}

The second lemma is an SoS triangle inequality, capturing the consequences of separation of the means.
The proof is standard given Fact~\ref{fact:sos-triangle}.
\begin{lemma}
  \label{lem:mixture-sos-triangle}
  Let $a,b \in \R^k$ and $t \in \N$ be a power of $2$.
  Let $\Delta = \|a -b\|$.
  Let $u = (u_1,\ldots,u_k)$ be indeterminates.
  Then $\proves{t} \|a - u\|^t + \|b - u\|^t \geq 2^{-t} \cdot \Delta^t$.
\end{lemma}

The last lemma helps put the previous two together.
Although we have phrased this lemma to concorde with the mixture model setting, we note that the proof uses nothing about mixture models and consists only of generic manipulations of pseudodistributions.
\begin{lemma}
  \label{lem:nice-pe}
  Let $\mu_1,\ldots,\mu_k,\cD_1,\ldots,\cD_k,X_1,\ldots,X_n$ be as in Theorem~\ref{thm:mixture-main}.
  Let $a_j$ be the $0/1$ indicator for the set of samples drawn from $\cD_j$.
  Suppose $\pE$ is a degree-$O(t)$ pseudodistribution which satisfies
  \begin{align*}
    \iprod{a_j,w} & \leq n\\
    \iprod{a_\ell,w} & \leq n\\
    \|\mu - \bmu_j\|^{2t} + \|\mu - \bmu_\ell\|^{2t} & \geq A \\
    \iprod{a_j,w}^t \|\mu - \bmu_j\|^{2t} & \leq B n \iprod{a_j,w}^{t-1} \|\mu - \bmu_j\|^{t}\\
    \iprod{a_\ell,w}^t \|\mu - \bmu_\ell\|^{2t} & \leq B n \iprod{a_\ell,w}^{t-1} \|\mu - \bmu_\ell\|^{t}
  \end{align*}
  for some scalars $A,B \geq 0$.
  Then
  \[
    \pE \iprod{a_j,w}\iprod{a_\ell,w} \leq \frac{2 n^2 B}{\sqrt A}\mper
  \]
\end{lemma}

Now we have the tools to prove Lemma~\ref{lem:mixture-ident-main}.
\begin{proof}[Proof of Lemma~\ref{lem:mixture-ident-main}]
  We will verify the conditions to apply Lemma~\ref{lem:nice-pe}.
  By Lemma~\ref{lem:mixture-ident-dist}, when (D1) holds, the pseudoexpectation $\pE$ satisfies
  \[
    \iprod{a_j,w}^t \|\mu - \bmu_j\|^{2t} & \leq B n \iprod{a_j,w}^{t-1} \|\mu - \bmu_j\|^{t}
  \]
  for $B = 4 (4t)^{t  / 2} /k$, and similarly with $j,\ell$ interposed.
  Similarly, by separation of the empirical means, $\pE$ satisfies $\|\mu - \bmu_j\|^{2t} + \|\mu - \bmu_\ell\|^{2t} \geq A$ for $A = 2^{-2t} \Delta^{2t}$, recalling that the empirical means are pairwise separated by at least $\Delta - 2 \Delta^{-t}$.
  Finally, clearly $\cA \proves{O(1)} \iprod{a_j,w} \leq n$ and similarly for $\iprod{a_\ell,w}$.
  So applying Lemma~\ref{lem:nice-pe} we get
  \[
    \pE \iprod{a_j,w}\iprod{a_\ell,w} \leq \frac{2 n^2 B}{\sqrt A} \leq \frac{n^2 2^{2t+2} t^{t/2} }{k} \cdot \frac 1 {\Delta^t}\mper\qedhere
  \]
\end{proof}

\subsection{Proof of Lemma~\ref{lem:mixture-ident-dist}}

In this subsection we prove Lemma~\ref{lem:mixture-ident-dist}.
We use the following helpful lemmata.
The first bounds error from samples selected from the wrong cluster using the moment inequality.

\begin{lemma}
\label{lem:mixture-ident-dist-1}
  Let $j,\cA,X_1,\ldots,X_n,\mu_j,\bmu_j$ be as in Lemma~\ref{lem:mixture-ident-dist}.
  Then
  \[
    \cA \proves{O(t)} \Paren{\sum_{i \in S_j} w_i \iprod{\mu - X_i, \mu - \bmu_j}}^t \leq 2 t^{t / 2} \cdot \iprod{a_j,w}^{t-1} \|\mu - \bmu_j\|^t\mper
  \]
\end{lemma}
\begin{proof}
  The proof goes by H\"older's inequality followed by the moment inequality in $\cA$.
  Carrying this out, by Fact~\ref{fact:sos-holder} and evenness of $t$,
  \[
    \{ w_i^2 = w_i \} \proves{O(t)}  \Paren{\sum_{i \in S_j} w_i \iprod{\mu - X_i, \mu - \bmu_j}}^t \leq \Paren{\sum_{i \in S_j} w_i}^{t-1} \cdot \Paren{\sum_{i \in [n]} w_i \iprod{\mu - X_i, \mu - \bmu_j}^t}\mper
  \]
  Then, using the main inequality in $\cA$,
  \[
    \cA \proves{O(t)} \Paren{\sum_{i \in S_j} w_i}^{t-1} \cdot 2 t^{t / 2} \cdot \|\mu - \bmu_j\|^t =  2 t^{t/2} \cdot \iprod{a_j,w}^{t-1} \|\mu - \bmu_j\|^t\mper\qedhere
  \]
\end{proof}

The second lemma bounds error from deviations in the empirical $t$-th moments of the samples from the $j$-th cluster.

\begin{lemma}
\label{lem:mixture-ident-dist-2}
  Let $\mu_1,\ldots,\mu_k, \cD_1,\ldots,\cD_k$ be as in Theorem~\ref{thm:mixture-main}.
  Suppose condition (D1) holds for samples $X_1,\ldots,X_n$.
  Let $w_1,\ldots,w_n$ be indeterminates.
  Let $u = u_1,\ldots,u_d$ be an indeterminate.
  Then for every $j \in [k]$,
  \[
    \{ w_i^2 = w_i \} \proves{O(t)} \Paren{\sum_{i \in S_j} w_i \iprod{X_i - \bmu_j, u}}^t \leq \iprod{a_j,w}^{t-1} \cdot 2 \cdot \frac nk \cdot \|u\|^t\mper
  \]
\end{lemma}
\begin{proof}
  The first step is H\"older's inequality again:
  \begin{align*}
    \{ w_i^2 = w_i \} \proves{O(t)} \Paren{\sum_{i \in S_j} w_i \iprod{X_i - \bmu_j, u}}^t \leq \iprod{a_j,w}^{t-1} \cdot \sum_{i \in S_j}\iprod{X_i - \bmu_j, u}^t\mper
  \end{align*}
  Finally, condition (D1) yields
  \[
    \{ w_i^2 = w_i \} \proves{O(t)} \Paren{\sum_{i \in S_j} w_i \iprod{X_i - \bmu_j, u}}^t\leq \iprod{a_j,w}^{t-1} \cdot 2 \cdot \frac nk \cdot \|u\|^t\mper \qedhere
  \]
\end{proof}

We can prove Lemma~\ref{lem:mixture-ident-dist} by putting together Lemma~\ref{lem:mixture-ident-dist-1} and Lemma~\ref{lem:mixture-ident-dist-2}.

\begin{proof}[Proof of Lemma~\ref{lem:mixture-ident-dist}]
  Let $j \in [k]$ be a cluster and recall $a_j \in \{0,1\}^n$ is the $0/1$ indicator for the samples in cluster $j$.
  Let $S_j$ be the samples in the $j$-th cluster, with empirical mean $\bmu_j$.
  We begin by writing $\iprod{a_j,w} \|\mu - \bmu_j\|^2$ in terms of samples $X_1,\ldots,X_n$.
  \begin{align*}
    \iprod{a_j,w} \|\mu - \bmu_j\|^2 & = \sum_{i \in [n]} w_i \iprod{\mu - \bmu_j, \mu - \bmu_j}\\
    & = \sum_{i \in S_j} w_i \iprod{\mu - X_i, \mu - \bmu_j} + \sum_{i \in [n]} w_i \iprod{X_i - \bmu_j, \mu - \bmu_j}\mper
  \end{align*}
  Hence, using $(a+b)^t \leq 2^t (a^t + b^t)$, we obtain
  \[
  \proves{O(t)} \iprod{a_j, w}^t \|\mu - \bmu_j\|^{2t} \leq 2^t \cdot \Paren{\sum_{i \in S_j} w_i \iprod{\mu - X_i, \mu - \bmu_j}}^t + 2^t \cdot \Paren{\sum_{i \in S_j} w_i \iprod{X_i - \bmu_j, \mu - \bmu_j}}^t\mper
  \]
  Now using Lemma~\ref{lem:mixture-ident-dist-1} and Lemma~\ref{lem:mixture-ident-dist-2},
  \[
    \cA \proves{O(t)} \iprod{a_j, w}^t \|\mu - \bmu_j\|^{2t} \leq 2^{t+2} t^{t / 2} \cdot \frac nk \cdot \iprod{a_j,w}^{t-1} \cdot \|\mu - \bmu_j\|^t
  \]
  as desired.
\end{proof}

\subsection{Proof of Lemma~\ref{lem:nice-pe}}
We prove Lemma~\ref{lem:nice-pe}.
The proof only uses standard SoS and pseudodistribution tools.
The main inequality we will use is the following version of H\"older's inequality.

\begin{fact}[Pseudoexpectation H\"older's, see Lemma A.4 in \cite{DBLP:conf/stoc/BarakKS14}]
\label{fact:pe-holders}
  Let $p$ be a degree-$\ell$ polynomial.
  Let $t \in N$ and let $\pE$ be a degree-$O(t\ell)$ pseudoexpectation on indeterminates $x$.
  Then
  \[
    \pE p(x)^{t-2} \leq \Paren{\pE p(x)^t}^{\tfrac{t-2}t}\mper
  \]
\end{fact}

Now we can prove Lemma~\ref{lem:nice-pe}.
\begin{proof}[Proof of Lemma~\ref{lem:nice-pe}]
  We first establish the following inequality.
  \begin{align}
    \pE \iprod{a_j,w}^t \iprod{a_\ell,w}^t  \|\mu - \bmu_j\|^{2t} \leq B^2 n^2 \cdot \pE \iprod{a_j,w}^{t-2} \iprod{a_\ell,w}^t\mper\label{eq:nice-1}
  \end{align}
  (The inequality will also hold by symmetry with $j$ and $\ell$ exchanged.)
  This we do as follows:
  \begin{align*}
  \pE \iprod{a_j,w}^t \iprod{a_\ell,w}^t \|\mu - \bmu_j\|^{2t} & \leq  Bn \pE \iprod{a_j,w}^{t-1} \iprod{a_\ell,w}^t \|\mu - \bmu_j\|^{2t}\\
  & \leq Bn \Paren{\pE \iprod{a_j,w}^{t-2} \iprod{a_\ell,w}^t }^{1/2} \cdot \Paren{ \pE \iprod{a_j,w}^t \iprod{a_\ell,w}^t \|\mu - \bmu_j\|^{2t}}^{1/2}
  \end{align*}
  where the first line is by assumption on $\pE$ and the second is by pseudoexpectation Cauchy-Schwarz.
  Rearranging gives the inequality \eqref{eq:nice-1}.

  Now we use this to bound $\pE \iprod{a_j,w}^t \iprod{a_\ell,w}^t$.
  By hypothesis,
  \[
  \pE \iprod{a_j,w}^t \iprod{a_\ell,w}^t \leq \frac 1 {A} \pE \iprod{a_j,w}^t \iprod{a_\ell,w}^t (\|\mu - \bmu_j\|^{2t} + \|\mu - \bmu_\ell\|^{2t})\mcom
  \]
  which, followed by \eqref{eq:nice-1} gives
  \[
    \pE \iprod{a_j,w}^t \iprod{a_\ell,w}^t \leq \frac 1 {A} \cdot B^2 n^2 \cdot \pE \Brac{\iprod{a_j,w}^{t-2} \iprod{a_\ell,w}^t + \iprod{a_\ell,w}^{t-2} \iprod{a_j,w}^t}\mper
  \]
  Using $\iprod{a_j,w},\iprod{a_\ell,w} \leq n$, we obtain
  \[
     \pE \iprod{a_j,w}^t \iprod{a_\ell,w}^t \leq \frac 2 {A} \cdot B^2 n^4 \cdot \pE \iprod{a_j,w}^{t-2} \iprod{a_\ell,w}^{t-2}\mper
  \]
  Finally, using Fact~\ref{fact:pe-holders}, the right side is at most $2B^2n^4/A \cdot \Paren{\pE \iprod{a_j,w}^t \iprod{a_\ell,w}^t}^{(t-2)/t}$, so cancelling terms we get
  \[
    \Paren{\pE \iprod{a_j,w}^t \iprod{a_\ell,w}^t}^{2/t} \leq \frac{2B^2n^4} {A}\mper
  \]
  Raising both sides to the $t/2$ power gives
  \[
  \pE \iprod{a_j,w}^t \iprod{a_\ell,w}^t \leq \frac{2^{t/2} B^{t} n^{2t}}{A^{t/2}}\mcom
  \]
  and finally using Cauchy-Schwarz,
  \[
    \pE \iprod{a_j,w} \iprod{a_\ell,w} \leq \Paren{\pE \iprod{a_j,w}^t \iprod{a_\ell,w}^t}^{1/t} \leq  \frac{2 n^2 B}{\sqrt A}\mper\qedhere
  \]
\end{proof}

\subsection{Rounding}
\label{sec:rounding}
In this section we state and analyze our second-moment round algorithm.
As have discussed already, our SoS proofs in the mixture model setting are quite strong, meaning that the rounding algorithm is relatively naive.

The setting in this section is as follows.
Let $n,m \in \N$ with $m \ll n$.
There is a ground-truth partition of $[n]$ into $m$ parts $S_1,\ldots,S_m$ such that $|S_i| = (1 \pm \delta) \tfrac nm$.
Let $A \in \R^{n \times n}$ be the $0/1$ indicator matrix for this partition, so $A_{ij} = 1$ if $i,j \in S_\ell$ for some $\ell$ and is $0$ otherwise.
Let $M \in \R^{n \times n}$ be a matrix such that $\|M - A\| \leq \e n$, where $\| \cdot \|$ is the Frobenious norm.
The algorithm takes $M$ and outputs a partition $C_1,\ldots,C_m$ of $[m]$ which makes few errors compared to $S_1,\ldots,S_m$.

\begin{algorithm}[htb]{}
\begin{algorithmic}[1]
\Function{RoundSecondMoments}{$M \in \R^{n \times n}, E \in \R$}
\State Let $S = [n]$
\State Let $v_1,\ldots,v_n$ be the rows of $M$
\For{$\ell = 1, \ldots, m$}
	\State Choose $i \in S$ uniformly at random
	\State Let 
	\[
	C_\ell = \left\{ i' \in S: \| v_i - v_{i'} \|_2 \leq 2 \frac{n^{1/2}}{E} \right\}
	\]
	\State Let $S \gets S \setminus C_\ell$
\EndFor
\State \textbf{return} The clusters $C_1,\ldots,C_m$.
\EndFunction
\end{algorithmic}
\caption{Rounding the second moment of $\pE [ww^\top]$}
\label{alg:round-second-moments}
\end{algorithm}
\noindent
We will prove the following theorem.
\begin{theorem}
\label{thm:round-main}
  With notation as before Algorithm~\ref{alg:round-second-moments} with $E = m$, with probability at least $1 - \e^2 m^3$ Algorithm~\ref{alg:round-second-moments} returns a partition $C_1,\ldots,C_m$ of $[n]$ such that (up to a permutation of $[m]$), $C_\ell = T_\ell \cup B_\ell$, where $T_\ell \subseteq S_\ell$ has size $|T_\ell| \geq |S_\ell| - \e^2 m n$ and $|B_\ell| \leq \e^2 m n$.
\end{theorem}

To get started analyzing the algorithm, we need a definition.
\begin{definition}
For cluster $S_j$, let $a_j \in \R^n$ be its $0/1$ indicator vector.
If $i \in S_j$, we say it is $E$-\emph{good} if $\|v_i - a_j\|_2 \leq \sqrt{n/E}$, and otherwise $E$-\emph{bad}, where $v_i$ is the $i$-th row of $M$.
Let $I_g \subseteq [n]$ denote the set of $E$-good indices and $I_b$ denote the set of $E$-bad indices.
(We will choose $E$ later.)
For any $j = 1, \ldots, k$, let $I_{g, j} = I_g \cap S_j$ denote the set of good indices from cluster $j$.
\end{definition}
\noindent
We have:
\begin{lemma}
\label{lem:if-choose-good}
Suppose $E$ as in \textsc{RoundSecondMoments} satisfies $E \geq m / 8$.
Suppose that in iterations $1, \ldots, m$, \textsc{RoundSecondMoments} has chosen only good vectors.
Then, there exists a permutation $\pi: [m] \to [m]$ so that $C_{\ell} = I_{g, \pi(\ell)} \cup B_{\ell}$, where $B_{\ell} \subseteq I_b$ for all $\ell$.
\end{lemma}
\begin{proof}
We proceed inductively.
We first prove the base case.
WLOG assume that the algorithm picks $v_1$, and that $v_1$ is good, and is from component $j$.
Then, for all $i \in I_{g, j}$, by the triangle inequality we have $\| v_i - v_1 \|_2 \leq 2 \frac{n^{1/2}}{B}$, and so $I_{g, j} \subseteq C_1$.
Moreover, if $i \in I_{g, j'}$ for some $j' \neq j$, we have
\[
  \| v_i - v_1 \|_2 \geq \| a_j' - a_j \|_2 - 2 \frac{n^{1/2}}{E^{1/2}} \geq \frac{n^{1/2}}{\sqrt m} - 2 \frac{n^{1/2}}{E^{1/2}} \;  > 2 \frac{n^{1/2}}{E^{1/2}},
\]
and so in this case $i \not\in C_1$.
Hence $C_1 =  I_{g, j} \cup B_1$ for some $B_1 \subseteq I_b$.

Inductively, suppose that if the algorithm chooses good indices in iterations $1, \ldots, a-1$, then there exist distinct $j_1, \ldots, j_{a - 1}$ so that $C_\ell = I_{g, j_\ell} \cup B_\ell$ for $B_\ell \subseteq I_b$.
We seek to prove that if the algorithm chooses a good index in iteration $a$, then $C_a = I_{g, j_a} \cup B_a$ for some $j_a \not\in \{j_1, \ldots, j_{a-1} \}$ and $B_a \subseteq I_b$.
Clearly by induction this proves the Lemma.
WLOG assume that the algorithm chooses $v_1$ in iteration $a$.
Since by assumption $1$ is good, and we have removed $I_{g_\ell}$ for $\ell = 1, \ldots, a - 1$, then $1 \in I_{g, j_a}$ for some $j_a \not\in \{j_1, \ldots, j_{a-1} \}$.
Then, the conclusion follows from the same calculation as in the base case.
\end{proof}

\begin{lemma}
\label{lem:few-bad}
  There are at most $\e^2 E n$ indices which are $E$-bad; i.e. $|I_b| \leq \e^2 E n$.
\end{lemma}
\begin{proof}
We have
\begin{align*}
\e^2n^2  &\geq \left\| M - \sum_{i \leq m} a_i a_i^\top \right\|_F^2 \geq \sum_{j} \sum_{i \in S_j \text{ bad} } \| v_i - a_j \|_2^2 \\
&\geq \frac{n}{E} |I_b| \; ,
\end{align*}
from which the claim follows by simplifying.
\end{proof}
\noindent
This in turns implies:
\begin{lemma}
\label{lem:choose-good-indices}
  With probability at least $1 - \e^2 m^3$, the algorithm \textsc{RoundSecondMoments} chooses good indices in all $k$ iterations.
\end{lemma}
\begin{proof}
By Lemma~\ref{lem:few-bad}, in the first iteration the probability that a bad vector is chosen is at most $\e^2 E$.
Conditioned on the event that in iterations $1, \ldots, a$ the algorithm has chosen good vectors, then by Lemma \ref{lem:if-choose-good}, there is at least one $j_a$ so that no points in $I_{g,j_a}$ have been removed.
Thus at least $(1-\delta)n/m$ vectors remain, and in total there are at most $\e^2 E n$ bad vectors, by Lemma~\ref{lem:few-bad}.
So, the probability of choosing a bad vector is at most $\e^2 E m$.
Therefore, by the chain rule of conditional expectation and our assumption , the probability we never choose a bad vector is at least
\[
\Paren{1 - \e^2 E m}^m
\]
Choosing $E = m$ this is $(1 - \e^2 m^2)^m \geq 1 - \e^2 m^3$.
as claimed.
\end{proof}

Now Theorem~\ref{thm:round-main} follows from putting together the lemmas.


\section{Robust estimation: algorithm and analysis}
\label{sec:robust}

Our algorithm for robust estimation is very similar to our algorithm for mixture models.
Suppose the underlying distribution $\cD$, whose mean $\mu^*$ the algorithm robustly estimates, is $10t$-explicitly bounded.
As a reminder, the input to the algorithm is a list of $X_1,\ldots,X_n \in \R^d$ and a sufficiently-small $\e >0$.
The guarantee is that at least $(1-\e)n$ of the vectors were sampled according to $\cD$, but $\e n$ of the vectors were chosen adversarially.

The algorithm solves a semidefinite program to obtain a degree $O(t)$ pseudodistribution which satisfies the system $\cA$ from Section~\ref{sec:overview-A} with $\alpha = 1-\e$ and $\tau = 0$.
Throughout this section, we will always assume that $\cA$ is instantiated with these parameters, and omit them for conciseness.
Then the algorithm just outputs $\pE \mu$ as its estimator for $\mu^*$.

Our main contribution in this section is a formal description of an algorithm \textsc{EstimateMean} which makes these ideas rigorous, and the proof of the following theorem about its correctness:
\begin{theorem}
  \label{thm:robust-main}
  Let $\eps > 0$ sufficiently small and $t \in \N$.
  Let $\cD$ be a $10t$-explicitly bounded distribution over $\R^d$ with mean $\mu^*$.
  Let $X_1, \ldots, X_n$ be an $\eps$-corrupted set of samples from $\cD$ where $n = d^{O(t)} / \eps^2$.
  Then, given $\eps, t$ and $X_1, \ldots, X_n$, the algorithm \textsc{EstimateMean} runs in time $d^{O(t)}$ and outputs $\mu$ so that $\| \mu - \mu^* \|_2 \leq O(t^{1/2} \eps^{1 - 1 / t})$, with probability at least $1 - 1 / d$.
\end{theorem}
\noindent
As a remark, observe that if we set $t = 2 \log 1/\eps$, then the error becomes $O(\eps \sqrt{\log 1 / \eps})$.
Thus, with $n = O(d^{O(\log 1 / \eps)} / \eps^2)$ samples and $n^{O(\log 1/\eps)} = d^{O(\log 1/\eps)^2}$ runtime, we achieve the same error bounds for general explicitly bounded distributions as the best known polynomial time algorithms achieve for Gaussian mean estimation.

\subsection{Additional Preliminaries}

Throughout this section, let $[n] = S_g \cup S_b$, where $S_g$ is the indices of the uncorrupted points, and $S_b$ is the indices of the corrupted points, so that $|S_b| = \eps n$ by assumption.
Moreover, let $Y_1, \ldots, Y_n$ be iid from $\cD$ so that $Y_i = X_i$ for all $i \in S_g$.

We now state some additional tools we will require in our algorithm.
\paragraph{Naive Pruning}
We will require the following elementary pruning algorithm, which removes all points which are very far away from the mean.
We require this only to avoid some bit-complexity issues in semidefinite programming; in particular we just need to ensure that the vectors $X_1,\ldots,X_n$ used to form the SDP have polynomially-bounded norms.
Formally:
\begin{lemma}[Naive pruning]
\label{lem:naive-pruning}
Let $\eps, t, \mu^*$, and $X_1, \ldots, X_n$ be as in Theorem~\ref{thm:robust-main}.
There is an algorithm \textsc{NaivePrune}, which given $\eps, t$ and $X_1, \ldots, X_n$, runs in time $O(\eps d n^2)$, and outputs a subset $S \subseteq [n]$ so that with probability $1 - 1 / d^{10}$, the following holds:
\begin{itemize}
\item No uncorrupted points are removed, that is $S_g \subseteq S$, and
\item For all $i \in S$, we have $\| X_i - \mu^* \| \leq O(d)$.
\end{itemize}
In this case, we say that \textsc{NaivePrune} succeeds.
\end{lemma}
\noindent
This algorithm goes by straightforward outlier-removal.
It is very similar the procedure described in Fact 4.18 of \cite{DBLP:conf/focs/DiakonikolasKK016} (using bounded $t$-th moments instead of sub-Gaussianity), so we omit it.

\paragraph{Satisfiability}
In our algorithm, we will use the same set of polynomial equations $\cAhat$ as in Lemma~\ref{lem:general-structured-subset-polys}.
However, the data we feed in does not exactly fit the assumptions in the Lemma.
Specifically, because the adversary is allowed to remove an $\eps$-fraction of good points, the resulting uncorrupted points are no longer iid from $\cD$.
Despite this, we are able to specialize Lemma~\ref{lem:general-structured-subset-polys} to this setting:
\begin{lemma}
\label{lem:corrupted-structured-subset-polys}
Fix $\eps > 0$ sufficiently small, and let $t \in \N,  t \geq 4$ be a power of $2$.
Let $\cD$ be a $10t$-explicitly bounded distribution.
Let $X_1, \ldots, X_n \in \R^d$ be an $\eps$-corrupted set of samples from $\cD$, and let $\cAhat$ be as in Lemma~\ref{lem:general-structured-subset-polys}.
The conclusion (1 -- Satisfiability) of Lemma~\ref{lem:general-structured-subset-polys} holds, with $w$ taken to be the $0/1$ indicator of the $(1 - \e)n$ good samples among $X_1,\ldots,X_n$.
\end{lemma}
\noindent
We sketch the proof of Lemma~\ref{lem:corrupted-structured-subset-polys} in Section~\ref{sec:corrupted-polys-proof}.

\subsection{Formal Algorithm Specification}
With these tools in place, we can now formally state the algorithm.
The formal specification of this algorithm is given in Algorithm~\ref{alg:estimate-mean}.

\begin{algorithm}[htb]{}
\begin{algorithmic}[1]
\Function{EstimateMean}{$\e,t,\kappa,X_1,\ldots,X_n$}
\State Preprocess: let $X_1,\ldots,X_n \gets \textsc{NaivePrune}(\e,X_1,\ldots,X_n)$, and let $\muhat$ be the empirical mean
\State Let $X_i \gets X_i - \muhat$
\State By semidefinite programming, find a pseudoexpectation of degree $O(t)$ which satisfies the structured subset polynomials from Lemma~\ref{lem:corrupted-structured-subset-polys}, with $\alpha = (1-\e)n$ and $\tau = 0$.
\State \textbf{return} $\pE \mu + \muhat$.
\EndFunction
\end{algorithmic}
\caption{Robust Mean Estimation}
\label{alg:estimate-mean}
\end{algorithm}

The first two lines of Algorithm~\ref{alg:estimate-mean} are only necessary for bit complexity reasons, since we cannot solve SDPs exactly.
However, since we can solve them to doubly-exponential accuracy in polynomial time, it suffices that all the quantities are at most polynomially bounded (indeed, exponentially bounded suffices) in norm, which these two lines easily achieve.
For the rest of this section, for simplicity of exposition, we will ignore these issues.

\subsection{Deterministic conditions}
With these tools in place, we may now state the deterministic conditions under which our algorithm will succeed.
Throughout this section, we will condition on the following events holding simultaneously:
\begin{enumerate}
\item[(E1)] \textsc{NaivePrune} succeeds, \label{enum:pruning}
\item[(E2)] The conclusion of Lemma~\ref{lem:corrupted-structured-subset-polys} holds, \label{enum:robust-A}
\item[(E3)] We have the following concentration of the uncorrupted points:
\[
\left\| \frac{1}{n} \sum_{i \in S_g} X_i - \mu^* \right\| \leq O(t^{1/2} \eps^{1 - 1 / t}) \; ,~\mbox{and}
\] \label{enum:robust-mean-conc}
\item[(E4)] We have the following concentration of the empirical $t$-th moment tensor:
\[
\frac{1}{n} \sum_{i \in [n]} \Brac{(Y_i - \mu^*)^{\otimes t / 2}} \Brac{ (Y_i - \mu^*)^{\otimes t / 2}}^\top \preceq \E_{X \sim \cD} \Brac{(X - \mu^*)^{\otimes t/2}} \Brac{(X - \mu^*)^{\otimes t/2}}^\top + 0.1 \cdot \Id \; ,
\]
for $\Id$ is the ${d^{t/2} \times d^{t / 2}}$-sized identity matrix.\label{enum:robust-tensor-conc}
\end{enumerate}
The following lemma says that with high probability, these conditions hold simultaneously:
\begin{lemma}
\label{lem:robust-deterministic-conditions}
Let $\eps, t, \mu^*$, and $X_1, \ldots, X_n \in \R^d$ be as in Theorem~\ref{thm:robust-main}.
Then, Conditions (E1)-(E4) hold simultaneously with probability at least $1 - 1 / d^{5}$.
\end{lemma}
\noindent
We defer the proof of this lemma to the Appendix.

For simplicity of notation, throughout the rest of the section, we will assume that \textsc{NaivePrune} does not remove any points whatsoever.
Because we are conditioning on the event that it removes no uncorrupted points, it is not hard to see that this is without loss of generality.

\subsection{Identifiability}
Our main identifiability lemma is the following.
\begin{lemma}
  \label{lem:robust-main-overview}
  Let $\eps, t, \mu^*$ and $X_1,\ldots,X_n \in \R^d$ be as in Theorem~\ref{thm:robust-main}, and suppose they satisfy (E1)--(E4).
  Then, we have
  \[
    \cA \proves{O(t)} \|\mu - \mu^*\|^{2t} \leq O(t^{t/2}) \cdot \e^{t-1} \cdot  \|\mu - \mu^* \|^t\mper
  \]
\end{lemma}

Since this lemma is the core of our analysis for robust estimation, in the remainder of this section we prove it.
The proof uses the following three lemmas to control three sources of error in $\pE \mu$, which we prove in Section~\ref{sec:robust-lemma-proofs}.
The first, Lemma \ref{lem:robust-good-samples} controls sampling error from true samples from $\cD$.

\begin{lemma}
\label{lem:robust-good-samples}
  Let $\eps, t, \mu^*$ and $X_1,\ldots,X_n \in \R^d$ be as in Theorem~\ref{thm:robust-main}, and suppose they satisfy (E1)--(E4) satisfy (E1)--(E4).
  Then, we have
  \[
    \proves{O(t)} \Paren{\sum_{i \in S_g} \iprod{X_i - \mu^*, \mu - \mu^*}}^t \leq O (\eps^{t - 1}) \cdot t^{t / 2} \cdot n^t \cdot \|\mu -\mu^*\|^t\mper
  \]
\end{lemma}

  To describe the second and third error types, we think momentarily of $w \in \R^n$ as the $0/1$ indicator for a set $S$ of samples whose empirical mean will be the output of the algorithm.
  (Of course this is not strictly true, but this is a convenient mindset in constructing SoS proofs.)
  The second type of error comes from the possible failure of $S$ to capture some $\e$ fraction of the good samples from $\cD$.
  Since $\cD$ has $O(t)$ bounded moments, if $T$ is a set of $m$ samples from $\cD$, the empirical mean of any $(1-\e)m$ of them is at most $\e^{1-1/t}$-far from the true mean of $\cD$.

\begin{lemma}
  \label{lem:robust-missed-good-samples}
  Let $\eps, t, \mu^*$ and $X_1,\ldots,X_n \in \R^d$ be as in Theorem~\ref{thm:robust-main}, and suppose they satisfy (E1)--(E4).
  Then, we have
  \[
    \cA \proves{O(t)} \Paren{\sum_{i \in S_g}(w_i-1) \iprod{X_i - \mu^*, \mu - \mu^*}}^t \leq 2 \eps^{t-1} n^t \cdot t^{t / 2} \cdot \|\mu - \mu^*\|^t\mper
  \]
\end{lemma}
\noindent
The third type of error is similar in spirit: it is the contribution of the original uncorrupted points that the adversary removed.
Formally:
\begin{lemma}
\label{lem:robust-removed-good-samples}
  Let $\eps, t, \mu^*$ and $X_1,\ldots,X_n \in \R^d$ and $Y_1, \ldots, Y_n \in \R^d$ be as in Theorem~\ref{thm:robust-main}, and suppose they satisfy (E1)--(E4).
  Then, we have
  \[
    \cA \proves{O(t)} \Paren{\sum_{i \in S_b} \iprod{Y_i - \mu^*, \mu - \mu^*}}^t \leq 2 \eps^{t-1} n^t \cdot t^{t / 2} \cdot \|\mu - \mu^*\|^t  \mper
  \]
\end{lemma}
\noindent
  Finally, the fourth type of error comes from the $\e n$ adversarially-chosen vectors.
  We prove this lemma by using the bounded-moments inequality in $\cA$.

\begin{lemma}
\label{lem:robust-bad-samples}
   Let $\eps, t, \mu^*$ and $X_1,\ldots,X_n \in \R^d$ be as in Theorem~\ref{thm:robust-main}, and suppose they satisfy (E1)--(E4).
   Then, we have
  \[
    \cA \proves{O(t)} \Paren{\sum_{i \notin S_g}w_i \iprod{X_i - \mu^*, \mu - \mu^*}}^t \leq 2 \eps^{t-1} n^t \cdot t^{t/2} \cdot \|\mu - \mu^*\|^t \mper
  \]
\end{lemma}

\noindent
With these lemmas in place, we now have the tools to prove Lemma~\ref{lem:robust-main-overview}.

\begin{proof}[Proof of Lemma~\ref{lem:robust-main-overview}]
Let $Y_1, \ldots, Y_n \in \R^d$ be as in Theorem~\ref{thm:robust-main}.
  We expand the norm $\|\mu - \mu^* \|^2$ as $\iprod{\mu - \mu^*, \mu - \mu^*}$ and rewrite $\sum_{i \in [n]} w_i \mu$ as $\sum_{i \in [n]} w_i X_i$:
  \begin{align}
    \sum_{i \in [n]} w_i \|\mu - \mu^*\|^2 &\stackrel{(a)}{=} \sum_{i \in [n]} w_i \iprod{X_i - \mu^*, \mu - \mu^*} \nonumber \\
    &\stackrel{(b)}{=} \sum_{i \in S_g} w_i \iprod{X_i - \mu^*, \mu - \mu^*} + \sum_{i \in S_b} w_i \iprod{X_i - \mu^*, \mu - \mu^*} \nonumber \\
    &\stackrel{(c)}{=} \sum_{i \in S_g} \iprod{X_i - \mu^*, \mu - \mu^*} + \sum_{i \in S_g}(w_i-1) \iprod{X_i - \mu^*, \mu - \mu^*} \nonumber \\
    &\qquad+ \sum_{i \in S_b} w_i \iprod{X_i - \mu^*, \mu - \mu^*} \nonumber \\
    &\stackrel{(d)}{=} \sum_{i \in [n]} \iprod{X_i - \mu^*, \mu - \mu^*} + \sum_{i \in S_g}(w_i-1) \iprod{X_i - \mu^*, \mu - \mu^*} \nonumber \\
    &\qquad-\sum_{i \in S_b} \iprod{Y_i - \mu^*, \mu - \mu^*} + \sum_{i \in S_b} w_i \iprod{X_i - \mu^*, \mu - \mu^*} \nonumber \; ,
  \end{align}
  where (a) follows from the mean axioms, (b) follows from splitting up the uncorrupted and the corrupted samples, (c) follows by adding and subtracting $1$ to each term in $S_g$, and (d) follows from the assumption that $Y_i = X_i$ for all $i \in [n]$.
  We will rearrange the last term by adding and subtracting $\mu$.
  Note the following polynomial identity:
  \[
    \iprod{X_i - \mu^*, \mu - \mu^*} = \iprod{X_i - \mu, \mu - \mu^*} + \|\mu - \mu^*\|^2
  \]
  and put it together with the above to get
  \begin{align*}
    \sum_{i \in [n]} w_i \|\mu - \mu^*\|^2 & = \sum_{i \in S_g} \iprod{X_i - \mu^*, \mu - \mu^*}+  \sum_{i \in S_g}(w_i-1) \iprod{X_i - \mu^*, \mu - \mu^*} \\
    &\qquad-\sum_{i \in S_b} \iprod{Y_i - \mu^*, \mu - \mu^*} +  \sum_{i \in S_b} w_i \iprod{X_i - \mu, \mu - \mu^*} +  \sum_{i \in S_b} w_i \|\mu - \mu^*\|^2 \mper
  \end{align*}
  which rearranges to
  \begin{align*}
  \sum_{i \in S_g} w_i \|\mu - \mu^*\|^2 &= \sum_{i \in S_g} \iprod{X_i - \mu^*, \mu - \mu^*}+  \sum_{i \in S_g}(w_i-1) \iprod{X_i - \mu^*, \mu - \mu^*} \\
    &\qquad-\sum_{i \in S_b} \iprod{Y_i - \mu^*, \mu - \mu^*} +  \sum_{i \in S_b} w_i \iprod{X_i - \mu, \mu - \mu^*} \mper
  \end{align*}
  Now we use $\proves{t} (x + y + z + w)^t \leq \exp(t) \cdot (x^t + y^t + z^t + w^t)$ for any even $t$, and Lemma \ref{lem:robust-good-samples}, Lemma~\ref{lem:robust-missed-good-samples}, and Lemma~\ref{lem:robust-bad-samples} and simplify to conclude
  \[
    \cA \proves{O(t)} \Paren{\sum_{i \in S_g} w_i}^t \|\mu - \mu^*\|^{2t} \leq \exp(t) \cdot t^{t / 2} \cdot n^t \cdot \e^{t-1} \cdot  \|\mu - \mu^* \|^t\mper
  \]
  Lastly, since $\cA \proves{2} \sum_{i \in T} w_i \geq (1 - 2\e)n$, we get
  \[
    \cA \proves{O(t)} \|\mu - \mu^*\|^{2t} \leq \exp(t) \cdot t^{t / 2} \cdot \e^{t-1} \cdot  \|\mu - \mu^* \|^t \; ,
  \]
  as claimed.
\end{proof}

\subsection{Rounding}
The rounding phase of our algorithm is extremely simple.
If $\pE$ satisfies $\cA$, we have by Lemma~\ref{lem:robust-main-overview} and pseudoexpectation Cauchy-Schwarz that
\[
\pE \|\mu - \mu^*\|^{2t} \leq  \exp(t) \cdot t^{t/2} \cdot \eps^{t - 1} \cdot \pE \Paren{ \| \mu - \mu^* \|^t} \leq \exp(t) \cdot t^{t/2} \cdot \eps^{t - 1} \cdot \pE \Paren{ \| \mu - \mu^* \|^{2t}}^{1/2}
\]
which implies that
\begin{equation}
\label{eq:robust-lift}
\pE \|\mu - \mu^*\|^{2t} \leq \exp (t) \cdot t^t \cdot \eps^{2(t - 1)}\; .
\end{equation}
Once this is known, analyzing $\| \pE \mu - \mu^*\|$ is straightforward.
By \eqref{eq:robust-lift} and pseudo-Cauchy-Schwarz again,
\[
  \|\pE[\mu] - \mu^*\|^2 \leq \pE \|\mu - \mu^*\|^2 \leq \Paren{\pE \|\mu - \mu^*\|^{2t}}^{1/t} \leq O (t \cdot \e^{2 - 2/t} ) \; ,
\]
which finishes analyzing the algorithm.

\subsection{Proofs of Lemmata~\ref{lem:robust-good-samples}--\ref{lem:robust-bad-samples}}
\label{sec:robust-lemma-proofs}

We first prove Lemma~\ref{lem:robust-good-samples}, which is a relatively straightforward application of SoS Cauchy Schwarz.
\begin{proof}[Proof of Lemma~\ref{lem:robust-good-samples}]
We have
\begin{align*}
\proves{O(t)} \Paren{\sum_{i \in S_g} \iprod{X_i - \mu^*, \mu - \mu^*}}^t &= \Paren{\Iprod{\sum_{i \in S_g} (X_i - \mu^*), \mu - \mu^*}}^t \\
&\leq \left\| \sum_{i \in S_g} (X_i - \mu^*) \right\|^t \| \mu - \mu^* \|^t \\
&\leq \Paren{n \cdot O \Paren{\eps^{1 - 1 / t}} \cdot t^{1/2}}^t \| \mu - \mu^* \|^t \; ,
\end{align*}
where the last inequality follows from (E3).
This completes the proof.
\end{proof}

Before we prove Lemmata~\ref{lem:robust-missed-good-samples}--\ref{lem:robust-bad-samples}, we prove the following lemma which we will use repeatedly:
\begin{lemma}
\label{lem:tensor-spectral-bound}
  Let $\eps, t, \mu^*$ and $Y_1, \ldots, Y_n \in \R^d$ be as in Theorem~\ref{thm:robust-main}, and suppose they satisfy (E4).
  Then, we have
  \[
  \cA \proves{O(t)} \sum_{i \in [n]} \iprod{Y_i - \mu^*, \mu - \mu^*}^t \leq 2 n t^{t / 2} \| \mu - \mu^* \|^t \mper
  \]
\end{lemma}
\begin{proof}
We have that
\begin{align*}
&\proves{t}\sum_{i \in [n]} \iprod{Y_i - \mu^*, \mu - \mu^*}^t = \Brac{(\mu - \mu^*)^{\otimes 2}}^\top \sum_{i \in [n]} \Brac{(Y_i - \mu^*)^{\otimes t / 2}} \Brac{ (Y_i - \mu^*)^{\otimes t / 2}}^\top \Brac{(\mu - \mu^*)^{\otimes 2}} \\
&\qquad \stackrel{(a)}{\leq} n \left( \Brac{(\mu - \mu^*)^{\otimes 2}}^\top \Paren{ \E_{X \sim \cD} \Brac{(X - \mu^*)^{\otimes t/2}} \Brac{(X - \mu^*)^{\otimes t/2}}^\top + 0.1 \cdot \Id } \Brac{(\mu - \mu^*)^{\otimes 2}} \right)\\
&\qquad = n \cdot \E_{X \sim \cD} \iprod{X - \mu^*, \mu - \mu^*}^t + n \cdot 0.1 \cdot \| \mu - \mu^* \|^t \\
&\qquad \stackrel{(b)}{\leq} 2n \cdot t^{t / 2} \| \mu - \mu^* \|^t \; ,
\end{align*}
where (a) follows from (E4) and (b) follows from $10t$-explicitly boundedness.
\end{proof}

\noindent
We now return to the proof of the remaining Lemmata.
\begin{proof}[Proof of Lemma~\ref{lem:robust-missed-good-samples}]
  We start by applying H\"older's inequality, Fact~\ref{fact:sos-holder}, (implicitly using that $w_i^2 = w_i \proves{2} (1-w_i)^2 = 1-w_i$), to get
  \begin{align*}
  \cA \proves{O(t)} & \Paren{\sum_{i \in S_g}(w_i-1) \iprod{X_i - \mu^*, \mu - \mu^*}}^t = \Paren{\sum_{i \in S_g}(1-w_i) \iprod{X_i - \mu^*, \mu - \mu^*}}^t\\
  & \leq \Paren{\sum_{i \in S_g} (w_i-1)}^{t-1} \Paren{\sum_{i \in S_g} \iprod{X_i - \mu^*, \mu - \mu^*}^t}\mper
  \end{align*}
  By Lemma~\ref{lem:tensor-spectral-bound}, we have 
  \begin{align*}
  \cA \proves{O(t)} \sum_{i \in S_g} \iprod{X_i - \mu^*, \mu - \mu^*}^t &\leq \sum_{i \in [n]} \iprod{Y_i - \mu^*, \mu - \mu^*}^t \\
  &\leq 2 n \cdot t^{t / 2}  \cdot \|\mu - \mu^*\|^t \mper
  \end{align*}
  At the same time,
  \[
    A \proves{2} \sum_{i \in T} (1-w_i) = (1-\e) n - \sum_{i \in [n]} w_i + \sum_{i \notin T} w_i = \sum_{i \notin T} w_i \leq \e n\mper
  \]
  So putting it together, we have
  \[
    \cA \proves{O(t)} \Paren{\sum_{i \in T}(w_i-1) \iprod{X_i - \mu^*, \mu - \mu^*}}^t \leq 2 (\e n)^{t-1} \cdot n \cdot t^{t / 2} \cdot \|\mu - \mu^*\|^t \; ,
  \]
  as claimed.
\end{proof}

\begin{proof}[Proof of Lemma~\ref{lem:robust-removed-good-samples}]
We apply H\"{o}lder's inequality to obtain that
\begin{align*}
\proves{O(t)} \Paren{\sum_{i \in S_b} \iprod{X_i - \mu^*, \mu - \mu^*}}^t &\leq |S_b|^{t - 1} \sum_{i \in S_b} \iprod{Y_i - \mu^*, \mu - \mu^*}^t \\
&\stackrel{(a)}{\leq} (\eps n)^{t - 1} \sum_{i \in [n]} \iprod{Y_i - \mu^*, \mu - \mu^*}^t \\
&\stackrel{(b)}{\leq} 2 (\eps n)^{t - 1} n t^{t / 2} \| \mu - \mu^* \|^t \; ,
\end{align*}
where (a) follows from the assumption on the size of $S_b$ and since the additional terms in the sum are SoS, and (b) follows follows from Lemma~\ref{lem:tensor-spectral-bound}.
This completes the proof.
\end{proof}

\begin{proof}[Proof of Lemma~\ref{lem:robust-bad-samples}]
  The proof is very similar to the proof of the two previous lemmas, except that we use the moment bound inequality in $\cA$.
  Getting started, by H\"older's:
  \[
    \cA \proves{O(t)} \Paren{\sum_{i \in S_b}w_i \iprod{X_i - \mu, \mu - \mu^*}}^t \leq \Paren{\sum_{i \in S_b} w_i}^{t-1} \Paren{\sum_{i \in S_b} w_i \iprod{X_i - \mu, \mu - \mu^*}^t}
  \]
  By evenness of $t$,
  \[
   \proves{t} \sum_{i \in S_b} w_i \iprod{X_i - \mu, \mu - \mu^*}^t \leq \sum_{i \in [n]} w_i \iprod{X_i - \mu, \mu - \mu^*}^t\mper
  \]
  Combining this with the moment bound in $\cA$,
  \[
   \cA \proves{O(t)} \Paren{\sum_{i \in S_b}w_i \iprod{X_i - \mu, \mu - \mu^*}}^t \leq \Paren{\sum_{i \in S_b} w_i}^{t-1} \cdot 2 \cdot t^{t / 2} \cdot n \cdot \|\mu - \mu^*\|^t\mper
  \]
  Finally, clearly $\cA \proves{2} \sum_{i \notin T} w_i \leq \e n$, which finishes the proof.
\end{proof}


\section{Encoding structured subset recovery with polynomials}
\label{sec:moment-polys}

The goal in this section is to prove Lemma~\ref{lem:general-structured-subset-polys}.
The eventual system $\cAhat$ of polynomial inequalities we describe will involve inequalities among matrix-valued polynomials.
We start by justifying the use of such inequalities in the SoS proof system.

\subsection{Matrix SoS proofs}
Let $x = (x_1,\ldots,x_n)$ be indeterminates.
We describe a proof system which can reason about inequalities of the form $M(x) \succeq 0$, where $M(x)$ is a symmetric matrix whose entries are polynomials in $x$.

Let $M_1(x),\ldots,M_m(x)$ be symmetric matrix-valued polynomials of $x$, with $M_i(x) \in \R^{s_i \times s_i}$, and let $q_1(x),\ldots,q_m(x)$ be scalar polynomials.
(If $s_i = 1$ then $M_i$ is a scalar valued polynomial.)
Let $M(x)$ be another matrix-valued polynomial.
We write
\[
  \{ M_1 \succeq 0,\ldots,M_m \succeq 0, q_1(x)=0,\ldots,q_m(x)=0 \} \proves{d} M \succeq 0
\]
if there are vector-valued polynomials $\{r_{S}^j\}_{j \leq N, S \subseteq [m]}$ (where the $S$'s are multisets), a matrix $B$, and a matrix $Q$ whose entries are polynomials in the ideal generated by $q_1,\ldots,q_m$, such that
\[
  M = B^\top \Brac{\sum_{S \subseteq [m]} \Paren{\sum_j (r_S^j(x))(r_S^j(x))^\top } \otimes \Brac{\otimes_{i \in S} M_i(x)}} B + Q(x)
\]
and furthermore that $\deg \Paren{\sum_j (r_S^j(x))(r_S^j(x))^\top } \otimes \Brac{\otimes_{i \in S} M_i(x)} \leq d$ for every $S \subseteq [m]$, and $\deg Q \leq d$.
Observe that in the case $M_1,\ldots,M_m,M$ are actually $1 \times 1$ matrices, this reduces to the usual notion of scalar-valued sum of squares proofs.

Adapting pseudodistributions to the matrix case,
we say a pseudodistribution $\pE$ of degree $2d$ satisfies the inequalities $\{ M_1(x) \succeq 0, \ldots, M_m(x) \succeq 0 \}$ if for every multiset $S \subseteq [m]$ and $p \in \R[x]$ such that $\deg \Brac{p(x)^2 \cdot (\otimes_{i \in S} M_i(x))} \leq 2d$,
\[
  \pE \Brac{p(x)^2 \cdot (\otimes_{i \in S} M_i(x))} \succeq 0\mper
\]

For completeness, we prove the following lemmas in the appendix.
\begin{lemma}[Soundness]
  Suppose $\pE$ is a degree-$2d$ pseudodistribution which satisfies constraints $\{M_1 \succeq 0,\ldots,M_m \succeq 0\}$, and
  \[
    \{M_1 \succeq 0,\ldots,M_m \succeq 0\} \proves{2d} M \succeq 0\mper
  \]
  Then $\pE$ satisfies $\{M_1 \succeq 0,\ldots,M_m \succeq 0, M\succeq 0 \}$.
\end{lemma}

\begin{lemma}
  \label{lem:matrix-sos-quad-form}
  Let $f(x)$ be a degree-$\ell$ $s$-vector-valued polynomial in indeterminates $x$.
  Let $M(x)$ be a $s \times s$ matrix-valued polynomial of degree $\ell'$.
  Then
  \[
    \{ M \succeq 0\} \proves{\ell \ell'} \iprod{f(x), M(x) f(x)} \geq 0\mper
  \]
\end{lemma}

Polynomial-time algorithms to find pseudodistributions satisfying matrix-SoS constraints follow similar ideas as in the non-matrix case.
In particular, recall that to enforce a scalar constraint $\{ p(x) \geq 0 \}$, one imposes the convex constraint $\pE p(x) (x^{\tensor d})(x^{\tensor d})^\top \succeq 0$.
Enforcing a constraint $\{ M(x) \succeq 0\}$ can be accomplished similarly by adding constraints of the form $\pE M(x) \succeq 0, \pE M(x) p(x) \succeq 0$, etc.

\subsection{Warmup: Gaussian moment matrix-polynomials}
In this section we develop the encoding as low degree polynomials of the following properties of an $n$-variate vector $w$ and a $d$-variate vector $\mu$.
We will not be able to encode exactly these properties, but they will be our starting point.
Let $d, n \in \N$, and suppose there are some vectors (a.k.a. samples) $X_1,\ldots,X_n \in \R^d$.
\begin{enumerate}
\item Boolean: $w \in \{0,1\}^n$.
\item Size: $(1 - \tau) \alpha n \leq \sum_{i \in [n]} w_i \leq (1 + \tau) \alpha n$.
\item Empirical mean: $\mu = \tfrac 1 {\sum_{i \in [n]} w_i} \sum_{i \in [n]} w_i X_i$.
\item $t$-th Moments: the $t$-th empirical moments of the vectors selected by the vector $w$, centered about $\mu$, are subgaussian.
That is,
\[
  \max_{u \in \R^d} \frac 1 {\alpha n} \sum_{i \in [n]} w_i \iprod{X_i - \mu, u}^t \leq 2 \cdot t^{t/2} \|u\|^t \mper
\]
\end{enumerate}
The second property is already phrased as two polynomial inequalities, and the third can be rearranged to a polynomial equation.
For the first, we use polynomial equations $w_i^2 = w_i$ for every $i \in [n]$.
The moment constraint will be the most difficult to encode.
We give two versions of this encoding: a simple one which will work when the distribution of the structured subset of samples to be recovered is Gaussian, and a more complex version which allows for any explicitly bounded distribution.
For now we describe only the Gaussian version.
We state some key lemmas and prove them for the Gaussian case.
We carry out the general case in the following section.

To encode the bounded moment constraint, for this section we let $M(w,\mu)$ be the following matrix-valued polynomial
\[
  M(w,\mu) = \frac 1 {\alpha n} \sum_{i \in [n]} w_i \Brac{\Paren{X_i - \mu}^{\otimes t/2}}\Brac{\Paren{X_i - \mu} ^{\otimes t/2}}^\top
\]

\begin{definition}[Structured subset axioms, Gaussian version]
For parameters $\alpha \in [0,1]$ (for the size of the subset), $t$ (for which empirical moment to control), and $\tau > 0$ (to account for some empirical deviations), the structured subset axioms are the following matrix-polynomial inequalities on variables $w = (w_1,\ldots,w_n), \mu = (\mu_1,\ldots,\mu_d)$.
  \begin{enumerate}
    \item booleanness: $w_i^2 = w_i$ for all $i \in [n]$
    \item size: $(1 - \tau) \alpha n \leq \sum_{i \in [n]} w_i \leq (1 + \tau) \alpha n$
    \item $t$-th moment boundedness: $ M(w,\mu) \preceq 2 \cdot \E_{X \sim \cN(0,\Id)} \Brac{X^{\tensor t/2}}\Brac{X^{\tensor t/2}}^\top$.
    \item $\mu$ is the empirical mean: $\mu \cdot \sum_{i \in [n]} w_i  = \sum_{i \in [n]} w_i X_i$.
  \end{enumerate}
\end{definition}
Notice that in light of the last constraint, values for the variables $\mu$ are always determined by values for the variables $w$, so strictly speaking $\mu$ could be removed from the program.
However, we find it notationally convenient to use $\mu$.
We note also that the final constraint, that $\mu$ is the empirical mean, will be used only for the robust statistics setting but seems unnecessary in the mixture model setting.

Next, we state and prove some key lemmas for this Gaussian setting, as warmups for the general setting.

\begin{lemma}[Satisfiability, Gaussian case]
  Let $d \in \N$ and $\alpha = \alpha(d) > 0$.
  Let $t \in \N$.
  Suppose $(1 - \tau) \alpha n \geq d^{100 t}$.
  Let $0.1 > \tau > 0$.
  If $X_1,\ldots,X_n \in \R^d$ has a subset $S \subseteq [n]$ such that $\{X_i\}_{i \in S}$ are iid samples from $\cN(\mu^*,\Id)$ and $|S| \geq (1 - \tau) \alpha n$, then with probability at least $1 - d^{-8}$ over these samples, the $\alpha,t,\tau$ structured subset axioms are satisfiable.
\end{lemma}
\begin{proof}
  Suppose $S$ has size exactly $(1 - \tau) \alpha n$; otherwise replace $S$ with a random subset of $S$ of size exactly $(\alpha - \tau)n$.
  As a solution to the polynomials, we will take $w$ to be the indicator vector of $S$ and $\mu = \frac 1 {|S|} \sum_{i \in [n]} w_i X_i$.
  The booleanness and size axioms are trivially satisfied.
  The spectral inequality
  \[
    \frac 1 {\alpha n} \sum_{i \leq [n]} w_i \Brac{\Paren{X_i - \mu}^{\otimes t/2}}\Brac{\Paren{X_i - \mu} ^{\otimes t/2}}^\top \preceq 2 \cdot \E_{X \sim \cN(0,\Id)} \Brac{X^{\tensor t/2}}\Brac{X^{\tensor t/2}}^\top
  \]
  follows from concentration of the empirical mean to the true mean $\mu^*$ and standard matrix concentration (see e.g. \cite{DBLP:journals/focm/Tropp12}).
\end{proof}

The next lemma is actually a corollary of Lemma~\ref{lem:matrix-sos-quad-form}.
\begin{lemma}[Moment bounds for polynomials of $\mu$, Gaussian case]
  Let $f(\mu)$ be a length-$d$ vector of degree-$\ell$ polynomials in indeterminates $\mu = (\mu_1,\ldots,\mu_k)$.
  The $t$-th moment boundedness axiom implies the following inequality with a degree $t\ell$ SoS proof.
  \begin{align*}
    & \left \{ M(w,\mu)\preceq 2 \cdot \E_{X \sim \cN(0,\Id)} \Brac{X^{\tensor t/2}}\Brac{X^{\tensor t/2}}^\top \right \}\\
    & \qquad \qquad \proves{O(t \ell)} \frac 1 {\alpha n} \sum_{i \in [n]} w_i \iprod{X_i - \mu, f(\mu)}^t \leq 2\cdot \E_{X \sim \cN(0,\Id)} \iprod{X,f(\mu)}^t\mper
  \end{align*}
\end{lemma}

\subsection{Moment polynomials for general distributions}
In this section we prove Lemma~\ref{lem:general-structured-subset-polys}.

We start by defining polynomial equations $\cAhat$, for which we introduce some extra variables
For every pair of multi-indices $\gamma,\rho$ over $[k]$ with degree at most $t/2$, we introduce a variable $M_{\gamma,\rho}$.
The idea is that $M = [ M_{\gamma,\rho} ]_{\gamma,\rho}$ forms an $n^{t/2} \times n^{t/2}$ matrix.
By imposing equations of the form $M_{\gamma,\rho} = f_{\gamma,\rho}(w,\mu)$ for some explicit polynomials $f_{\gamma,\rho}$ of degree $O(t)$, we can ensure that
\[
  \iprod{u^{\tensor t/2}, M u^{\tensor t/2}} = 2 \cdot t^{t/2} \|u\|^t - \frac 1 {\alpha n} \sum_{i \in [n]} w_i \iprod{X_i - \mu, u}^t\mper
\]
(This equation should be interpreted as an equality of polynomials in indeterminates $u$.)
Let $\cL$ be such a family of polynomial equations.
Our final system $\cAhat(\alpha,t,\tau)$ of polynomial equations and inequalities follows.
The important parameters are $\alpha$, controlling the size of the set of samples to be selected, and $t$, how many moments to control.
The parameter $\tau$ is present to account for random fluctuations in the sizes of the cluster one wants to recover.

\begin{definition}
  Let $\cAhat(\alpha,t,\tau)$ be the set of (matrix)-polynomial equations and inequalities on variables $w, \mu, M_{\gamma,\rho}$ containing the following.
  \begin{enumerate}
    \item Booleanness: $w_i^2 = w_i$ for all $i \in [n]$
    \item Size: $(1 - \tau) \alpha n \leq \sum w_i \leq (1 + \tau) \alpha n$.
    \item Empirical mean: $\mu \cdot \sum_{i \in [n]} w_i = \sum_{i \in [n]} w_i X_i$.
    \item The equations $\cL$ on $M$ described above.
    \item Positivity: $M \succeq 0$.
  \end{enumerate}
\end{definition}

In the remainder of this section we prove the satisfiability and moment bounds parts of Lemma~\ref{lem:general-structured-subset-polys}.
To prove the lemma we will need a couple of simple facts about SoS proofs.

\begin{fact}
  \label{fact:general-dist-1}
  Let $X_1,\ldots,X_m \in \R^d$.
    Let $v \in \R^d$ have $\|v\| \leq 1$.
  Let $Y_i = X_i + v$.
  Let $t \in \N$ be even.
  Suppose there is $C \in \R$ with $C \geq 1$ such that for all $s \leq t$,
  \[
    \frac 1m \sum_{i \in [m]} \|X_i\|^s \leq C^s
  \]
  Then
  \[
    \proves{t} \frac 1m \sum_{i \in [n]} \Brac{\iprod{X_i,u}^t - \iprod{Y_i,u}^t} \leq \Paren{2^{t} C^{t-1} \|v\|} \|u\|^t
  \]
  and similarly for $\tfrac 1m \sum_{i \in [n]} \Brac{\iprod{Y_i,u}^t - \iprod{X_i,u}^t}$.
\end{fact}
\begin{proof}
  Expanding $\iprod{Y_i,u}^t$, we get
  \[
    \iprod{Y_i,u}^t = \iprod{X_i+v,u}^t = \sum_{s \leq t} \binom ts \iprod{X_i,v}^s \iprod{v,u}^{t-s}\mper
  \]
  So,
  \[
    \frac 1m \sum_{i \in [m]} \Brac{\iprod{X_i,u}^t - \iprod{Y_i,u}^t} = - \frac 1m \sum_{i \in [m]} \sum_{s < t} \binom ts \iprod{X_i,u}^s \iprod{v,u}^{t-s}\mper
  \]
  For each term, by Cauchy-Schwarz, $\proves{t} \iprod{X_i,u}^s \iprod{v,u}^{t-s} \leq \|X_i\|^s\|v\|^{t-s} \cdot \|u\|^t$.
  Putting these together with the hypothesis on $\tfrac 1n \|X_i\|^s$ and counting terms finishes the proof.
\end{proof}

\begin{proof}[Proof of Lemma~\ref{lem:general-structured-subset-polys}: Satisfiability]
  By taking a random subset $S$ if necessary, we assume $|S| = (1 - \tau) \alpha n = m$.
  We describe a solution to the system $\cAhat$.
  Let $w$ be the $0/1$ indicator vector for $S$.
  Let $\mu = \tfrac 1m \sum_{i \in S} w_i X_i$.
  This satisfies the Boolean-ness, size, and empirical mean axioms.

  Describing the assignment to the variables $\{M_{\gamma,\rho}\}$ takes a little more work.
  Re-indexing and centering, let $Y_1 = X_{i_1} - \mu ,\ldots,Y_m = X_{i_m} - \mu$ be centered versions of the samples in $S$, where $S = \{i_1,\ldots,i_m\}$ and $\mu$ remains the empirical mean $\tfrac 1m \sum_{i \in S} X_i$.
  First suppose that the following SoS proof exists:
  \[
    \proves{t} \frac 1 {\alpha n} \sum_{i \in S} \iprod{Y_i,u}^t \leq 2 \cdot t^{t/2} \|u\|^t\mper
  \]
  Just substituting definitions, we also obtain
  \[
    \proves{t} \frac 1 {\alpha n} \sum_{i \in [n]} w_i \iprod{X_i - \mu, u}^t \leq 2 \cdot t^{t/2} \|u\|^t\mper
  \]
  \emph{where now $w$ and $\mu$ are scalars, not variables, and $u$ are the only variables remaining}.
  The existence of this SoS proof means there is a matrix $P \in \R^{d^{t/2} \times d^{t/2}}$ such that $P \succeq 0$ and
  \[
    \iprod{u^{\tensor t/2}, P u^{\tensor t/2}} = 2 t^{t/2} \|u\|^t - \frac 1 {\alpha n} \sum_{i \in [n]} w_i \iprod{X_i - \mu,u}^t\mper
  \]
  Let $M_{\gamma,\rho} = P_{\gamma,\rho}$.
  Then clearly $M \succeq 0$ and $M,w,\mu$ together satisfy $\cL$.

  It remains to show that the first SoS proof exists with high probability for large enough $m$.
  Since $t$ is even and $0.1 > \tau > 0$, it is enough to show that
  \[
    \proves{t} \frac 1m \sum_{i \in [S]} \iprod{Y_i,u}^t \leq 1.5 \cdot t^{t/2} \|u\|^t
  \]
  Let $Z_i = X_i - \mu^*$, where $\mu^*$ is the true mean of $\cD$.
  Let
  \[
    a(u) = \frac 1m \sum_{i \in S} \Brac{ \iprod{Z_i,u}^t - \iprod{Y_i,u}^t } \qquad b(u) = \frac 1m \sum_{i \in S} \iprod{Z_i,u}^t - \E_{Z \sim \cD - \mu^*} \iprod{Z,u}^t\mper
  \]
  We show that for $d \geq 2$,
  \[
    \proves{t} a(u) \leq \tfrac 1 4 \|u\|^t \qquad \proves{t} b(u) \leq \tfrac 1 4 \|u\|^t
  \]
  so long as the following hold
  \begin{enumerate}
    \item (bounded norms) for every $s \leq t$ it holds that $\tfrac 1m \sum_{i \in [m]} \|Z_i\|^s \leq s^{100s} d^{s/2}$. \label{itm:gendist-1}
    \item (concentration of empirical mean) $\|\mu - \mu^*\| \leq d^{-5t}$. \label{itm:gendist-2}
    \item (bounded coefficients) For every multiindex $\theta$ of degree $|\theta| = t$, one has $\Abs{\tfrac 1m \sum_{i \in [m]} Z_i^\theta - \E_{Z \sim \cD}Z^\theta} \leq d^{-10t}$.  \label{itm:gendist-3}
  \end{enumerate}
  We verify in Fact~\ref{fact:gendist-concentration} following this proof that these hold with high probability by standard concentration of measure, for $m \geq d^{100t}$ and $\cD$ $10t$-explicitly bounded, as assumed.
  Together with the assumption $\proves{t} \E_{Z \sim \cD - \mu^*} \iprod{Z,u}^t \leq t^{t/2} \|u\|^t$, this will conclude the proof.

  Starting with $a(u)$, using Fact~\ref{fact:general-dist-1}, it is enough that $2^t C^{t-1} \|v\| \leq \tfrac 14$, where $v = \mu - \mu^*$ and $C$ is such that $\tfrac 1m \sum_{i \in [m]} \|Z_i\|^s \leq C^s$.
  By \ref{itm:gendist-1} and \ref{itm:gendist-2}, we can assume $\|v\| \leq d^{-5t}$ and $C = t^{100} d^{1/2}$.
  Then the conclusion follows for $t \geq 3$.

  We turn to $b(u)$.
  A typical coefficient of $b(u)$ in the monomial basis---say, the coefficient of $u^\theta$ for some multiindiex $\theta$ of degree $|\theta| = t$, looks like
  \[
    \frac 1m \sum_{i \in [m]} Y_i^\theta - \E_{Y \sim \cD} Y^\theta\mper
  \]
  By assumption this is at most $d^{-10t}$ in magnitude, so the sum of squared coefficients of $b(u)$ is at most $d^{-18t}$.
  The bound on $b(u)$ for $d \geq 2$.
\end{proof}

\begin{proof}[Proof of Lemma~\ref{lem:general-structured-subset-polys}: Moment bounds]
  As in the lemma statement, let $f(\mu)$ be a vector of degree-$\ell$ polynomials in $\mu$.
  By positivity and Lemma~\ref{lem:matrix-sos-quad-form},
  \[
    M(w,\mu) \geq 0 \proves{O(t\ell)} \iprod{f(\mu)^{\tensor t/2}, M(w,\mu) f(\mu)^{\tensor t/2}} \geq 0\mper
  \]
  Using this in conjunction with the linear equations $\cL$,
  \[
    \cAhat \proves{O(t\ell)} 2 t^{t/2} \|f(\mu)\|^t - \frac 1 {\alpha n} \sum_{i \in [n]} w_i \iprod{X_i - \mu, f(\mu)}^t \geq 0
  \]
  which is what we wanted to show.
\end{proof}

\begin{fact}[Concentration for items \ref{itm:gendist-1}, \ref{itm:gendist-2},\ref{itm:gendist-3}]
\label{fact:gendist-concentration}
  Let $d,t \in \N$.
  Let $\cD$ be a mean-zero distribution on $\R^d$ such that $\E \iprod{Z,u}^s \leq s^s \|u\|^s$ for all $s \leq 10t$ for every $u \in \R^d$.
  Then for $t \geq 4$ and large enough $d$ and $m \geq d^{100t}$, for $m$ independent samples $Z_1,\ldots,Z_m \sim \cD$,
  \begin{enumerate}
    \item (bounded norms)  for every $s \leq t$ it holds that $\tfrac 1m \sum_{i \in [m]} \|Z_i\|^s \leq s^{100s} d^{s/2}$.
    \item (concentration of empirical mean) $\Norm{\tfrac 1m \sum_{i \in [m]} Z_i} \leq d^{-5t}$.
    \item (bounded coefficients) For every multiindex $\theta$ of degree $|\theta| = t$, one has $\Abs{\tfrac 1m \sum_{i \in [m]} Z_i^\theta - \E_{Z \sim \cD}Z^\theta} \leq d^{-10t}$.
  \end{enumerate}
\end{fact}
\begin{proof}
  The proofs are standard applications of central limit theorems, in particular the Berry-Esseen central limit theorem \cite{berry1941accuracy}, since all the quantities in question are sums of iid random variables with bounded moments.
  We will prove only the first statement; the others are similar.

  Note that $\tfrac 1m \sum_{i \in [m]} \|Z_i\|^s$ is a sum of iid random variables.
  Furthermore, by our moment bound assumption, $\E_{Z \sim \cD} \|Z\|^s \leq s^{2s} d^{s/2}$.
  We will apply the Berry-Esseen central limit theorem \cite{berry1941accuracy}.
  The second and third moments $\E (\|Z\|^s - \E\|Z\|^s)^{2}, \E (\|Z\|^s - \E\|Z\|^s)^{3}$ are bounded, respectively, as $s^{O(s)} k^{s}$ and $s^{O(s)} d^{3s/2}$.
  By Berry-Esseen,
  \[
  \Pr \left \{ \frac {\sqrt m}{d^{s/2}} \cdot \frac 1m \sum_{i \in [m]} \|Z_i\|^s > r + \frac{\sqrt m}{d^{s/2}} \E \|Z\|^s \right \} \leq 
  e^{-\Omega(r^2)} + s^{O(s)} \cdot m^{-1/2}\mper
  \]
\end{proof}

Finally we remark on the polynomial-time algorithm to find a pseudoexpectation satisfying $\cAhat$.
As per \cite{sos-notes-12}, it is just necessary to ensure that if $x = (w,\mu)$, the polynomials in $\cAhat$ include $\|x\|^2 \leq M$ for some large number $M$.
In our case the equation $\|x\|^2 \leq (nkm)^{O(1)}$ can be added without changing any arguments.

\subsection{Modifications for robust estimation}
\label{sec:corrupted-polys-proof}
We briefly sketch how the proof of Lemma~\ref{lem:general-structured-subset-polys} may be modified to prove Lemma~\ref{lem:corrupted-structured-subset-polys}.
The main issue is that $\cAhat$ of Lemma~\ref{lem:general-structured-subset-polys} is satisfiable when there exists an SoS proof
\[
  \proves{t} \frac 1 {(1-\eps)n} \sum_{i \in [n]} w_i \iprod{X_i - \mu, u}^t \leq 2 t^{t/2} \|u\|^t
\]
where $\mu$ is the empirical mean of $X_i$ such that $w_i = 1$.
In the proof of Lemma~\ref{lem:general-structured-subset-polys} we argued that this holds when $w$ is the indicator for a set of iid samples from a $10t$-explicitly bounded distribution $\cD$.
However, in the robust setting, $w$ should be taken to be the indicator of the $(1-\e)n$ good samples remaining from such a set of iid samples after $\eps n$ samples are removed by the adversary.
If $Y_1,\ldots,Y_n$ are the original samples, with empirical mean $\mu^*$, the proof of Lemma~\ref{lem:general-structured-subset-polys} (with minor modifications in constants) says that with high probability,
\[
\proves{t} \frac 1 n \sum_{i \in [n]} \iprod{Y_i - \mu^*, u}^t \leq 1.1 t^{t/2} \|u\|^t
\]
For small-enough $\eps$, this also means that
\[
  \proves{t} \frac 1{(1-\e)n} \sum_{i \text{ good}} \iprod{X_i - \mu^*, u}^t \leq 1.2 t^{t/2} \|u\|^t\mper
\]
This almost implies that $\cAhat$ is satisfiable given the $\e$-corrupted vectors $X_1,\ldots,X_n$ and parameter $\alpha = (1-\e)n$, except for that $\mu^* = \tfrac 1n \sum_{i \in [n]} Y_i$ and we would like to replace it with $\mu = \tfrac 1 {(1-\e)n} \sum_{i \text{ good}} X_i$.
This can be accomplished by noting that, as argued in Section~\ref{sec:robust}, with high probability $\|\mu - \mu^*\| \leq O(t \cdot \e^{1 - 1/t})$.

\section{Acknowledgements}
The authors would like to thank David Steurer, Daniel Freund, Guatam Kamath, Pablo Parillo, and especially Aravindan Vijayaraghavan for some helpful conversations.

\bibliographystyle{amsalpha}
\bibliography{mathreview,wiki,dblp,custom,scholar}

\appendix

\section{Toolkit for sum of squares proofs}
\label{sec:toolkit}
\begin{fact}[See Fact A.1 in \cite{DBLP:conf/focs/MaSS16} for a proof]
  \label{fact:sos-cs}
  Let $x_1,\ldots,x_n,y_1,\ldots,y_n$ be indeterminates.
  Then
  \[
    \proves{4} \Paren{\sum_{i \leq n} x_i y_i}^2 \leq \Paren{\sum_{i \leq n} x_i^2} \Paren{\sum_{i \leq n} y_i^2}\mper
  \]
\end{fact}

\begin{fact}
\label{fact:sos-triangle}
  Let $x,y$ be $n$-length vectors of indeterminates.
  Then
  \[
  \proves{2} \|x + y\|^2 \leq 2 \|x\|^2 + 2 \|y\|^2\mper
  \]
\end{fact}
\begin{proof}
  The sum of squares proof of Cauchy-Schwarz implies that $\|x\|^2 + \|y\|^2 - 2\iprod{x,y}$ is a sum of squares.
  Now we just expand
  \[
  \|x+y\|^2 = \|x\|^2 + \|y\|^2 + 2\iprod{x,y} \preceq 2(\|x\|^2 + \|y\|^2)\mper
  \]
\end{proof}

\begin{fact}
\label{fact:poly-moment-bound}
  Let $P(x) \in \R[x]_{\ell}$ be a homogeneous degree $\ell$ polynomial in indeterminates $x = x_1,\ldots,x_n$.
  Suppose that the coefficients of $P$ are bounded in $2$-norm:
  \[
    \sum_{\alpha \subseteq [n]} \hat{P}(\alpha)^2 \leq C\mper
  \]
  (Here $\hat{P}(\alpha)$ are scalars such that $P(x) = \sum_{\alpha} \hat{P}(\alpha) x^{\alpha}$.)
  Let $a,b \in \N$ be integers such that $a + b = \ell$.
  Then
  \[
    \proves{\max(2a,2b)} P(x) \leq \sqrt{C} ( \|x\|^{2a} + \|x\|^{2b} )\mper
  \]
\end{fact}
\begin{proof}
  Let $M$ be a matrix whose rows and columns are indexed by multisets $S \subseteq [n]$ of sizes $a$ and $b$.
  Thus $M$ has four blocks: an $(a,a)$ block, an $(a,b)$ block, a $(b,a)$ block, and a $(b,b)$ block.
  In the $(a,b)$ and $(b,a)$ blocks, put matrices $M_{ab}, M_{ba}$ such that $\iprod{x^{\tensor a}, M_{ab} x^{\tensor b}} = \tfrac 12 \mper P(x)$.
  In the $(a,a)$ and $(b,b)$ blocks, put $\sqrt C \cdot I$.
  Then, letting $z = (x^{\tensor a},x^{\tensor b})$, we get $\iprod{z, Mz} = \sqrt C( \|x\|^{2a} + \|x\|^{2b}) - P(x)$.
  Note that $\|M_{ab}\| \leq \sqrt{C}$ by hypothesis, so $M \succeq 0$, which completes the proof.
\end{proof}

\begin{fact}
\label{fact:sos-gaussian-moment}
  Let $u = (u_1,\ldots,u_k)$ be a vector of indeterminantes.
  Let $D$ be sub-Gaussian with variancy proxy $1$.
  Let $t \geq 0$ be an integer.
  Then we have
  \begin{align*}
  &\proves{2t} \E_{X \sim D} \iprod{X,u}^{2t} \leq (2t)! \cdot \|u\|^{2t} \\
  &\proves{2t} \E_{X \sim D} \iprod{X,u}^{2t} \geq -(2t)! \cdot \|u\|^{2t} \; .
  \end{align*}
\end{fact}
\begin{proof}
  Expand the polynomial in question.
  We have
  \[
  \E_{X \sim D} \iprod{X,u}^{2t} = \E_{X \sim D} \sum_{\beta} u^{\beta} \E [X^\beta] \; .
  \]
  Let $\beta$ range over $[k]^{2t}$
  \[
  \proves{2t}  \sum_{\beta} u^{2\beta} \E X^{2\beta} \leq (2t)! \sum_{\beta \text{ even }} u^{\beta} \leq \|u\|_2^{2t}\mper
  \]
  where we have used upper bounds on the Gaussian moments $\E X^{2\beta}$ and that every term is a square in $u$.
\end{proof}

\begin{fact}[SoS Cauchy-Schwarz (see Fact A.1 in \cite{DBLP:conf/focs/MaSS16} for a proof)]
  Let $x_1,\ldots,x_n,y_1,\ldots,y_n$ be indeterminates.
  Then
  \[
    \proves{4} \Paren{\sum_{i \leq n} x_i y_i}^2 \leq \Paren{\sum_{i \leq n} x_i^2} \Paren{\sum_{i \leq n} y_i^2}\mper
  \]
\end{fact}
\begin{fact}[SoS H\"older]
  \label{fact:sos-holder}
  Let $w_1,\ldots,w_n$ and $x_1,\ldots,x_n$ be indeterminates.
  Let $q \in \N$ be a power of $2$.
  Then
  \[
    \{ w_i^2 = w_i \, \forall i \in[n] \} \proves{O(q)} \Paren{\sum_{i \leq n} w_i x_i}^q \leq \Paren{\sum_{i \leq n} w_i}^{q-1} \cdot \Paren{ \sum_{i \leq n} x_i^q}
  \]
  and
  \[
    \{ w_i^2 = w_i \, \forall i \in[n] \} \proves{O(q)} \Paren{\sum_{i \leq n} w_i x_i}^q \leq \Paren{\sum_{i \leq n} w_i}^{q-1} \cdot \Paren{ \sum_{i \leq n} w_i \cdot x_i^q}\mper
  \]
\end{fact}
\begin{proof}
  We will only prove the first inequality. 
  The second inequality follows since $w_i^2 = w_i \proves{2} w_i x_i = w_i \cdot (w_i x_i)$, applying the first inequality, and observing that $w_i^2 = w_i \proves{q} w_i^q = w_i$.

  Applying Cauchy-Schwarz (Fact~\ref{fact:sos-cs}) and the axioms, we obtain to start that for any even number $t$,
  \begin{align*}
  \{ w_i^2 = w_i \, \forall i \in[n] \} \proves{O(t)} & \Brac{\Paren{\sum_{i \leq n} w_i x_i}^2}^{t/2} = \Brac{\Paren{\sum_{i\leq n} w_i^2 x_i}^2}^{t/2}\\
  & \leq \Brac{ \Paren{\sum_{i \leq n} w_i^2} \Paren{\sum_{i \leq n} w_i^2 x_i^2}}^{t/2} = \Paren{\sum_{i \leq n} w_i}^{t/2} \Paren{\sum_{i \leq n} w_i x_i^2}^{t/2} \mper
  \end{align*}
  It follows by indution that
  \[
   \{ w_i^2 = w_i \, \forall i \in[n] \} \proves{O(t)} \Brac{\Paren{\sum_{i \leq n} w_i x_i}}^q \leq \Paren{\sum_{i \leq n} w_i}^{q-2} \Paren{\sum_{i \leq n} w_i x_i^{q/2}}^2\mper
  \]
  Applying Fact~\ref{fact:sos-cs} one more time to get $\Paren{\sum_{i \leq n} w_i x_i^{q/2}} \leq \Paren{\sum_{i \leq n} w_i^2} \Paren{\sum_{i \leq n} x_i^q}$ and then the axioms $w_i^2 = w_i$ completes the proof.
\end{proof}

\subsection{Examples of explicitly bounded distributions}
\label{sec:explicitly-bounded-families}
In this section, we show that many natural high dimensional distributions are explicitly bounded.
Recall that if a univariate distribution $X$  \emph{sub-Gaussian} (with variancy proxy $\sigma$) with mean $\mu$ then we have the following bound on its even centered moments for $t \geq 4$:
\[
\E [(X - \mu)^{t}] \leq \sigma^t \Paren{\frac t2}^{t / 2} \; ,
\]
if $t$ is even.

More generally, we will say a univariate distribution is $t$-bounded with mean $\mu$ and variance proxy $\sigma$ if the following general condition holds for all even $4 \leq s \leq t$:
\[
\E [(X - \mu)^{s}] \leq \sigma^s \Paren{ \frac s2}^{s / 2} \; .
\]
The factor of $1/2$ in this expression is not important and can be ignored upon first reading.

Our main result in this section is that any rotation of products of independent $t$-bounded distributions with variance proxy $1/2$ is $t$-explicitly bounded with variance proxy $1$:
\begin{lemma}
Let $\cD$ be a distribution over $\R^d$ so that $\cD$ is a rotation of a product distribution $\cD'$ where each coordinate of $\cD$ is a $t$-bounded univariate distribution with variance proxy $1/2$.
Then $\cD$ is $t$-explicitly bounded (with variance proxy $1$).
\end{lemma}
\begin{proof}
Since the definition of explicitly bounded is clearly rotation invariant, it suffices to show that $\cD'$ is $t$-explicitly bounded.
For any vector of indeterminants $u$, and for any $4 \leq s \leq t$ even, we have
\begin{align*}
\proves{s} \E_{X \sim \cD'} \iprod{X - \mu, u}^s &= \E_{X \sim \cD'} \iprod{X - \E_{X' \sim \cD'} X', u}^s \\
&= \E_{X \sim \cD'} \Paren{\E_{X'} \iprod{X - X', u}}^s \\
&\leq \E_{X, X' \sim \cD'} \iprod{X - X', u}^s \; ,
\end{align*}
where $X'$ is an independent copy of $X$, and the last line follows from SoS Cauchy-Schwarz.
We then expand the resulting polynomial in the monomial basis:
\begin{align*}
\E_{X, X' \sim \cD'} \iprod{X - X', u}^s &= \sum_\alpha u^\alpha \E_{X, X'} (X - X')^\alpha \\
&= \sum_{\mbox{$\alpha$ even}} u^\alpha \E_{X, X'} (X - X')^\alpha \; ,
\end{align*}
since all $\alpha$ with odd monomials disappear since $X - X'$ is a symmetric product distribution.
By $t$-boundedness, all remaining coefficients are at most $s^{cs}$, from which we deduce
\[
\proves{s} \E_{X, X' \sim \cD'} \iprod{X - X', u}^s \leq s^{s / 2} \sum_{\mbox{$\alpha$ even}} u^\alpha = s^{s / 2} \| u \|^s \; ,
\]
which proves that $\cD'$ is $t$-explicitly bounded, as desired.
\end{proof}
\noindent
As a corollary observe this trivially implies that all Guassians $\normal (\mu, \Sigma)$ with $\Sigma \preceq I$ are $t$-explicitly bounded for all $t$.

We note that our results are tolerant to constant changes in the variancy proxy (just by scaling down).
In particular, this implies that our results immediately apply for all rotations of products of $t$-bounded distributions with a loss of at most $2$.

\section{Sum of squares proofs for matrix positivity -- omitted proofs}
\begin{lemma}[Soundness]
  Suppose $\pE$ is a degree-$2d$ pseudodistribution which satisfies constraints $\{M_1 \succeq 0,\ldots,M_m \succeq 0\}$, and
  \[
    \{M_1 \succeq 0,\ldots,M_m \succeq 0\} \proves{2d} M \succeq 0\mper
  \]
  Then $\pE$ satisfies $\{M_1 \succeq 0,\ldots,M_m \succeq 0, M\succeq 0 \}$.
\end{lemma}
\begin{proof}
  By hypothesis, there are $r_S^j$ and $B$ such that
  \[
    M = B^\top \Brac{\sum_{S \subseteq [m]} \Paren{\sum_j (r_S^j(x))(r_S^j(x))^\top } \otimes \Brac{\otimes_{i \in S} M_i(x)}} B\mper
  \]
  Now, let $T \subseteq [m]$ and $p$ be a polynomial.
  Let $M' = \tensor_{i \in T} M_i$.
  Suppose that $\deg \Paren{p^2 \cdot M \tensor M'} \leq 2d$.
  Using the hypothesis on $M$, we obtain
  \begin{align*}
  p^2 \cdot M \tensor M' & = p^2 \cdot B^\top \Brac{\sum_{S \subseteq [m]} \Paren{\sum_j (r_S^j(x))(r_S^j(x))^\top } \otimes \Brac{\otimes_{i \in S} M_i(x)}} B \tensor M'\\
  & = (B \tensor I)^\top \Brac{ p^2 \cdot \Brac{\sum_{S \subseteq [m]} \Paren{\sum_j (r_S^j(x))(r_S^j(x))^\top } \otimes \Brac{\otimes_{i \in S} M_i(x)}} \tensor M' } (B \tensor I)\mper
  \end{align*}
  Applying $\pE$ to the above, note that by hypothesis,
  \[
    \pE \Brac{ p^2 \cdot \Brac{\sum_{S \subseteq [m]} \Paren{\sum_j (r_S^j(x))(r_S^j(x))^\top } \otimes \Brac{\otimes_{i \in S} M_i(x)}} \tensor M'} \succeq 0\mper
  \]
  The lemma follows by linearity.
\end{proof}

\begin{lemma}
  Let $f(x)$ be a degree-$\ell$ $s$-vector-valued polynomial in indeterminates $x$.
  Let $M(x)$ be a $s \times s$ matrix-valued polynomial of degree $\ell'$.
  Then
  \[
    \{ M \succeq 0\} \proves{\ell \ell'} \iprod{f(x), M(x) f(x)} \geq 0\mper
  \]
\end{lemma}
\begin{proof}
  Let $u \in \R^{s \tensor s}$ have entries $u_{ij} = 1$ if $i = j$ and otherwise $u_{ij} = 0$.
  Then $\iprod{f(x), M(x) f(x)} = u^\top (M(x) \otimes f(x) f(x)^\top) u$.
\end{proof}


\section{Omitted Proofs from Section \ref{sec:robust}}
\label{sec:robust-proofs}
\subsection{Proof of Lemma \ref{lem:robust-deterministic-conditions}}

We will show that each event (E1)--(E4) holds with probability at least $1 - d^{-8}$.
Clearly for $d$ sufficiently large this implies the desired guarantee.
That (E1) and (E2) occur with probability $1 - d^{-8}$ follow from Lemmas~\ref{lem:naive-pruning} and \ref{lem:corrupted-structured-subset-polys}, respectively.
It now suffices to show (E3) and (E4) holds with high probability.
Indeed, that (E4) holds with probability $1 - d^{-8}$ follows trivially from the same proof of Lemma~\ref{lem:general-structured-subset-polys} (it is in fact a simpler version of this fact).

Finally, we show that (E3) holds.

By basic concentration arguments (see e.g. \cite{DBLP:journals/corr/abs-1011-3027}), we know that by our choice of $n$, with probability $1 - d^{-8}$ we have that
\begin{equation}
\label{eq:basic-mean-conc}
\left\| \frac{1}{n} \sum_{i \in [n]} X_i - \mu^* \right\| \leq \eps \; .
\end{equation}
Condition on the event that this and (E4) simultaneously hold.
Recall that $Y_i$ for $i = 1, \ldots, n$ are defined so that $Y_i$ are iid and $Y_i = X_i$ for $i \in S_g$.
By the triangle inequality, we have
\begin{align*}
\left\| \frac{1}{|S_g|} \sum_{i \in S_g} X_i - \mu^* \right\| &\leq \frac{n}{|S_g|} \left\| \frac{1}{n} \sum_{i \in [n]} Y_i - \mu^* \right\| + \frac{|S_b|}{|S_g|} \left\| \frac{1}{|S_b|} \sum_{i \in S_b} Y_i - \mu^* \right\| \\
&\stackrel{(a)}{\leq} \frac{\eps}{1 - \eps} + \frac{|S_b|}{|S_g|} \left\| \frac{1}{|S_b|} \sum_{i \in S_b} Y_i - \mu^* \right\| \; , \numberthis \label{eq:app-robust-1}
\end{align*}
where (a) follows from~\eqref{eq:basic-mean-conc}. 

We now bound the second term in the RHS.
For any unit vector $u \in \R^d$, by H\"{o}lder's inequality,
\begin{align*}
\Iprod{\sum_{i \in S_b} (Y_i - \mu^*), u}^t &\leq |S_b|^{t - 1} \sum_{i \in S_b} \Iprod{(Y_i - \mu^*), u}^t \\&
\leq |S_b|^{t - 1} \sum_{i \in [n]} \Iprod{(Y_i - \mu^*), u}^t \\
&= |S_b|^{t - 1} \Brac{u^{\otimes t / 2}}^\top \sum_{i \in [n]} \Brac{(Y_i - \mu^*)^{\otimes t / 2}} \Brac{(Y_i - \mu^*)^{\otimes t / 2}}^\top \Brac{u^{\otimes t / 2}} \\
&\stackrel{(a)}{\leq} |S_b|^{t - 1} \cdot n \cdot \Brac{u^{\otimes t / 2}}^\top \Paren{ \E_{Y \sim D} \Brac{(Y - \mu^*)^{\otimes t / 2}} \Brac{(Y - \mu^*)^{\otimes t / 2}}^\top + \delta \cdot \Id} \Brac{(Y - \mu^*)^{\otimes t / 2}} \\
&= |S_b|^{t - 1} \cdot n \cdot \Paren{\E_{Y \sim D} \Iprod{Y - \mu^*, u}^t + \delta} \\
&\leq |S_b|^{t - 1} \cdot n \cdot (t^{t / 2} + \delta) \\
&\stackrel{(b)}{\leq} 2 |S_b|^{t - 1} \cdot n \cdot t^{t / 2} \; ,
\end{align*}
where (a) follows from (E4), and (b) follows since $\delta \ll t^{t}$.
Hence 
\[
\left\| \sum_{i \in S_b} (Y_i - \mu^*) \right\| = \max_{\| u \| = 1} \Iprod{\sum_{i \in S_b} (Y_i - \mu^*), u} \leq O( |S_b|^{1 - 1/t} \cdot n^{1/t} \cdot t^{1/2})
\]
Taking the $t$-th root on both sides and combining it with~\eqref{eq:app-robust-1} yields
\[
\left\| \frac{1}{|S_g|} \sum_{i \in S_g} X_i - \mu^* \right\| \leq \frac{\eps}{1 - \eps} + \frac{\eps}{1 - \eps} (n / |S_b| )^{-1/t} \cdot t^{1/2} = O(\eps^{1 - 1 /t} \cdot t^{1/2}) \; ,
\]
as claimed.


\section{Mixture models with nonuniform weights}
\label{sec:nonuniform-mixtures}

In this section we describe at a high level how to adapt the algorithm given in Section~\ref{sec:mixture} to handle non-uniform weights.
We assume the mixture components now have mixture weights $\eta \leq \lambda_1 \leq \ldots \leq \lambda_k \leq 1$ where $\sum \lambda_i = 1$, where $\eta > 0$ is some fixed constant.
We still assume that all pairs of means satisfy $\| \mu_i - \mu_j \| \geq k^\gamma$ for all $i \neq j$.
In this section we describe an algorithm \textsc{LearnNonUniformMixtureModel}, and we sketch a proof of the following theorem concerning its correctness:
\begin{theorem}
\label{thm:mixture-main-nonuniform}
Let $\eta, \gamma > 0$ be fixed.
Let $\cD$ be a non-uniform mixture of $k$ distributions $\cD_1, \ldots, \cD_k$ in $\R^d$, where each $\cD_j$ is a $O(1 / \gamma)$-explicitly bounded distribution with mean $\mu_j$, and we have $\| \mu_i - \mu_j \| \geq k^{\gamma}$.
Furthermore assume that the smallest mixing weight of any component is at least $\eta$.
Then, given $X_1, \ldots, X_n$ iid samples from $\cD$ where $n \geq \frac{1}{\eta} (dk)^{O(1 / \gamma)}$, \textsc{LearnNonUniformMixtureModel} runs in $O(n^{1/t})$ time and outputs estimates $\muhat_1, \ldots, \muhat_m$ so that there is some permutation $\pi:[m] \to [m]$ so that $\| \muhat_i - \mu_{\pi(i)} \|_2 \leq k^{-10}$ with probability at least $1 - k^{-5}$.
\end{theorem}

Our modified algorithm is as follows: take $n$ samples $X_1, \ldots, X_n$ where $n$ is as in Theorem~\ref{thm:mixture-main-nonuniform}.
Then, do single-linkage clustering as before, and work on each cluster separately, so that we may assume without loss of generality that all means have pairwise $\ell_2$ distance at most $O(\poly (d, k))$.

Within each cluster, we do the following.
For $\alpha' = 1, 1 - \xi, 1 - 2 \xi, \ldots, \eta$ for $\xi = \poly (\eta / k)$, iteratively form $\cAhat$ with $\alpha = \alpha'$, $t = O \Paren{\frac{1}{\gamma} }$, and $\tau, \delta = k^{-10}$.
Attempt to find a pseudo-expectation $\pE$ that satisfies $\cAhat$ with these parameters with minimal $\| \pE ww^\top \|_F$.
If none exists, then retry with the next $\alpha'$.
Otherwise, a rounding algorithm on $\pE ww^\top$ to extract clusters.
Remove these points from the dataset, and then continue with the next $\alpha'$.

However, the rounding algorithm we require here is somewhat more involved than the naive rounding algorithm used previously for learning mixture models.
In particular, we no longer know exactly the Frobenius norm of the optimal solution: we cannot give tight upper and lower bounds.
This is because components with mixing weights which are just below the threshold $\alpha'$ may or may not contribute to the optimal solution that the SDP finds.
Instead, we develop a more involved rounding algorithm \textsc{RoundSecondMomentsNonuniform}, which we describe below.

Our invariant is that every time we have a feasible solution to the SDP, we remove at least one cluster (we make this more formal below).
Repeatedly run the SDP with this $\alpha'$ until we no longer get a feasible solution, and then repeat with a slightly smaller $\alpha'$.
After the loop terminates, output the empirical mean of every cluster.
The formal specification of this algorithm is given in Algorithm~\ref{alg:learn-mixture-nonuniform}.

\begin{algorithm}[htb]{}
\begin{algorithmic}[1]
\Function{LearnNonuniformMixtureMeans}{$t,\eta,X_1,\ldots,X_n$}
\State Let $\xi \gets \eta^2 / (dk)^{-100}$
\State Let $\mathcal{C} \gets \{ \}$, the empty set of clusters
\State Let $\cX \gets \{X_1, \ldots, X_n\}$

\State Perform naive clustering on $\cX$ to obtain $\cX_1, \ldots, \cX_\ell$.

\For{each $\cX_r$}

	\State Let $\alpha' \gets 1$

	\While{$\alpha'  \geq \eta - k^{-8}$}
		\State By semidefinite programming (see Lemma~\ref{lem:general-structured-subset-polys}, item \ref{itm:solve}), find a pseudoexpectation of degree $t = O( \frac{1}{\gamma} )$ which satisfies the structured subset polynomials from Lemma~\ref{lem:general-structured-subset-polys}, with $\alpha = \alpha' n,$ and $\delta, \tau = k^{-8}$ with data points as in $\cX$.
		\While{the SDP is feasible}
			\State Let $\pE$ be the pseudoexpectation returned
			\State Let $M \gets \pE ww^\top$.
			\State Run the algorithm \textsc{RoundSecondMomentsNonuniform} on $M$ to obtain a cluster $C$.
			\State Let $\mathcal{C} \gets \mathcal{C} \cup \{C\}$
			\State Remove all points in $C$ from $\cX_r$
		\EndWhile
		\State Let $\alpha' \gets \alpha' - \xi$
	\EndWhile
\EndFor
\State \textbf{return} The empirical mean of every cluster in $\mathcal{C}$
\EndFunction
\end{algorithmic}
\caption{Mixture Model Learning}
\label{alg:learn-mixture-nonuniform}
\end{algorithm}

For $j = 1, \ldots, k$ let $S_j$ be the set of indices of points in $X_1, \ldots, X_n$ which were drawn from $\cD_j$, and let $a_j \in \R^n$ be the indicator vectors for these sets as before.
Our key invariant is the following: for every $\alpha'$ such that the SDP returns a feasible solution, we must have $|\alpha' - \lambda_j| \leq O(\xi)$ for some $j$, and moreover, for every $j$ so that $\lambda_j \geq \alpha' + O(\xi)$, there must be exactly one cluster $C_\ell$ output by the algorithm at this point so that $|C_\ell \triangle S_j| \leq k^{-10} \poly (\eta) \cdot  \cdot n$.
Moreover, every cluster output so far must be of this form.
For any $\alpha'$, we say that the algorithm up to $\alpha'$ is \emph{well-behaved} if it satisfies this invariant for the loops in the algorithm for $\alpha''$ for $\alpha'' > \alpha'$.

It is not hard to show, via arguments exactly as in Section~\ref{sec:robust} and \ref{sec:moment-polys} that the remaining fraction of points from these components which we have not removed as well as the small fraction of points we have removed from good components do not affect the calculations, and so we will assume for simplicity in the rest of this discussion that we have removed all samples from components  $j$ with $\lambda_j \geq \alpha' + O(\xi)$.

\subsection{Sketch of proof of correctness of Algorithm \ref{alg:learn-mixture-nonuniform}}
Here we outline the proof of correctness of Algorithm \ref{alg:learn-mixture-nonuniform}.
The proof follows very similar ideas as the proof of correctness of Algorithm \ref{alg:learn-mixture}, and so for conciseness we omit many of the details.
As before, for simplicity assume that the naive clustering returns only one cluster, as otherwise we can work on each cluster separately, so that for all $i$, we have $\| \mu_i \| \leq O(\poly (d, k))$ after centering.

We now show why this invariant holds.
Clearly this holds at the beginning of the algorithm.
We show that if it holds at any step, it must also hold at the next time at which the SDP is feasible.
Fix such an $\alpha'$.
By assumption, we have removed almost all points from components $j$ with $\lambda_j \geq \alpha' + k^{-8}$, and have only removed a very small fraction of points not from these components.

By basic concentration, we have $\left| \lambda_j n - |S_j| \right| \leq o(n)$ for all $j$ except with negligble probability, and so for the rest of the section, for simplicity, we will slightly cheat and assume that $\lambda_j n = |S_j|$.
It is not hard to show that this also does not effect any calculations.

The main observation is that for any choice of $\alpha'$, by essentially same logic as in Section~\ref{sec:mixture}, we still have the following bound for all $i \neq j$ for an $\alpha'$ well-behaved run:
\begin{equation}
\label{eq:nonuniform-moment-bound}
\cAhat \proves{O(t)} \iprod{a_i, w} \iprod{a_j, w} \leq \frac{\eta n^2 t^{O(t)}}{k^{2 t \gamma}} = O(\eta \xi^2) \cdot (\alpha')^2 n^2 \; ,
\end{equation}
for $\cAhat$ instantiated with $\alpha = \alpha'$, where the last line follows by our choice of $t$ sufficiently large.

We now show this implies:
\begin{lemma}
\label{lem:no-weight-on-small}
With parameters as above, for any $\alpha'$ well-behaved run, we have $\cAhat \proves{O(t)} \iprod{a_i, w} \leq O(\xi^2) \cdot \alpha' n$ for any $j$ so that $\lambda_j n \leq (\alpha' - O(\xi^4)) n$.
\end{lemma}
\begin{proof}
We have
\begin{align*}
\cAhat \proves{t} \sum_{j' \neq j} \iprod{a_i, w} &= \alpha' n - \iprod{a_j, w} \geq \Omega (\xi^2) n \; ,
\end{align*}
and hence 
\begin{align*}
\cAhat \proves{O(t)} \Omega (\xi^2) n \iprod{a_i, w} &\leq \iprod{a_i, w} \sum_{j \neq i} \iprod{a_j, w} \\
&\leq \frac{1}{\eta} O(\eta \xi^4) \cdot (\alpha')^2 \cdot n^2 \; ,
\end{align*}
from which we deduce $\cAhat \proves{O(t)} \iprod{a_i, w} \leq O(\xi^2) \cdot \alpha' n$.
\end{proof}

We now show that under these conditions, there is an algorithm to remove a cluster:

\subsection{Rounding Well-behaved runs}

\begin{lemma}
\label{lem:rounding-main-nonuniform}
Let $\alpha', \eta, \gamma, t$ be as in Theorem~\ref{thm:mixture-main-nonuniform}.
Suppose that $\cAhat$ is satisfiable with this set of parameters, that the algorithm has been $\alpha'$ well-behaved, and~\eqref{eq:nonuniform-moment-bound} holds.
Then, there is an algorithm \textsc{RoundSecondMomentsNonuniform} which given $\pE$ outputs a cluster $C$ so that $|C \triangle S_j| \leq (\eta / dk)^{O(1)} n$ with probability $1 - (\eta / dk)^{O(1)}$.
\end{lemma}

Formally, let $v_i \in \R^n$  be so that for all $i, j$, we have $\iprod{v_i, v_j} = \pE w_i w_j$.
Such $v_i$ exist because $\pE ww^\top$ is PSD, and can be found efficiently via spectral methods.
For any cluster $j$, let $V_j$ denote the set of vectors $v_i$ for $i \in S_j$.

Our algorithm will proceed as follows: choose a random $v_i$ with $\| v_i \|^2 \geq \alpha' / 100$, and simply output as the cluster the set of $\ell$ so that $\| v_i - v_\ell \| \leq O(\sqrt{d \xi})$.

We now turn to correctness of this algorithm.
Define $T$ to be the set of clusters $j$ with $|\lambda_j - \alpha'| \leq O(\xi^4)$.
We first show:
\begin{lemma}
Assume that~\eqref{eq:nonuniform-moment-bound} holds.
Then
\[
\sum_{\ell \in T} \sum_{i, j \in S_\ell} \| v_i - v_j \|^2 \leq O(d^2 \xi^2) (\alpha')^2 n^2 \; .
\]
\end{lemma}
\begin{proof}
Observe that 
\begin{align*}
\sum_{\ell \in T} \sum_{i, j \in S_\ell} \| v_i - v_j \|^2 &= \sum_{\ell \in T} \sum_{i, j \in S_\ell} \| v_i \|^2 + \| v_j \|^2 - 2 \iprod{v_i, v_j} \\
&= \sum_{\ell \in T} \left( 2 |S_\ell| \sum_{i \in S_\ell} \| v_i \|^2 - 2 \sum_{i, j \in S_\ell} \iprod{v_i, v_j} \right) \; .
\end{align*}
By assumption, we have
\[
\sum_{\ell \in T} \sum_{i \in S_\ell} |S_\ell| \| v_\ell \|^2 &= (\alpha' \pm O(\xi^4)) n \sum_{\ell \in T} \| v_\ell \|^2 &= (\alpha' \pm O(\xi^4)) n \cdot \pE \left( \sum_{\ell \in T} \sum_{i \in S_\ell} w_i^2 \right) \; .
\]
Since by Lemma~\ref{lem:no-weight-on-small} we have $\pE [\sum_{\ell \not\in T} \sum_{i \in S_\ell} w_i^2] \leq d O(\xi^2) \alpha n$, we conclude that 
\[
\alpha n \geq \pE \left( \sum_{\ell \in T} \sum_{i \in S_\ell} w_i^2 \right) \geq (1 - d O(\xi^2)) \alpha' n \; .
\]
All of this allows us to conclude 
\[
\sum_{\ell \in T} \sum_{i \in S_\ell} |S_\ell| \| v_\ell \|^2 = (1 \pm O(d \xi^2)) (\alpha')^2 n^2 \; .
\]
On the other hand, we have 
\begin{align*}
\sum_{\ell \in T} \sum_{i, j \in S_\ell} \iprod{v_i, v_j} &= \sum_{\ell \in T} \pE \iprod{a_\ell, w}^2 \; ,
\end{align*}
but we have
\begin{align*}
(\alpha')^2 n^2 &= \pE \Paren{\sum_{\ell} \iprod{a_\ell, w}}^2 \\
&= \sum_{\ell \neq j} \pE [\iprod{a_\ell, w} \iprod{a_j, w}] + \sum_{\ell \not\in T} \iprod{a_\ell, w}^2 + \sum_{\ell \in T} \iprod{a_\ell, w}^2 \; .
\end{align*}

The first term is at most $O(d^2 \eta \xi^2) (\alpha')^2 n^2$ by~\eqref{eq:nonuniform-moment-bound} and the second term is at most $d O(\xi^2) \alpha' n$ by Lemma~\ref{lem:no-weight-on-small}, so overall we have that 
\[
\sum_{\ell \in T} \pE \iprod{a_\ell, w}^2 &= (1 \pm O(d^2 \xi^2))(\alpha')^2 n^2 \; .
\]
Hence putting it all together we have
\[
\sum_{\ell \in T} \sum_{i, j \in S_\ell} \| v_i - v_j \|^2 &= O(d^2 \xi^2) (\alpha')^2 n^2 \; ,
\]
as claimed.
\end{proof}
\noindent
As a simple consequence of this we have:
\begin{lemma}
Assume that~\eqref{eq:nonuniform-moment-bound} holds.
For all $\ell \in T$, there exists a ball $B$ of radius $O(\sqrt{d \xi})$ so that $|V_\ell \triangle B| \leq O(d \xi) \alpha' n$.
\end{lemma}
\begin{proof}
Suppose not, that is, for all $B$ with radius $O(d \xi)$, we have $|S_\ell \triangle B| \leq \Omega(d \xi) \alpha' n$.
Consider the ball of radius $O(\sqrt{m \xi})$ centered at each $v_i$ for $i \in S_\ell$.
By assumption there are $\Omega(d \xi) \alpha' n$ vectors outside the ball, that is, with distance at least $\Omega(\sqrt{d \xi})$ from $v_i$.
Then
\[
\sum_{i, j \in S_\ell} \| v_i - v_j \|_2^2 \geq n \cdot \Omega (d \xi) \Omega(d \xi) \alpha n \geq \Omega (d^2 \xi^2) \alpha' n \; ,
\]
which contradicts the previous lemma.
\end{proof}

Associate to each cluster $\ell \in T$ a ball $B_\ell$ so that $|V_\ell \triangle B| \leq \Omega(d \xi) \alpha' n$.
Let $\phi_\ell$ denote the center of $B_\ell$.
We now show that if we have two $j, \ell$ so that either $\| \phi_j \|$ or $\| \phi_\ell \|$ is large, then $B_\ell$ and $B_j$ must be disjoint.
Formally:
\begin{lemma}
Assume that~\eqref{eq:nonuniform-moment-bound} holds.
Let $j, \ell \in T$ so that $\| \phi_j \|^2 + \| \phi_\ell \|^2 \geq \Omega(\alpha')$ .
Then $B_j \cap B_\ell = \emptyset$.
\end{lemma}
\begin{proof}
We have 
\begin{align*}
\sum_{i \in B_j, k \in B_\ell} \| v_i - v_k \|^2 &= \sum_{i \in B_j, k \in B_\ell} \| v_i \|^2 + \|v_k\|^2 - 2 \langle v_i, v_k \rangle \\
&= |B_\ell| \sum_{i \in B_j} \| v_i \|^2 + |B_j| \sum_{k \in B_\ell} \| v_k \|^2 - 2 \sum_{i \in B_j, k \in B_\ell} \pE w_i w_k \\
&\geq (\alpha' - O(\xi^4)) n \left( \sum_{i \in B_j} \| v_i \|^2 + |B_j| \sum_{k \in B_\ell} \| v_k \|^2 \right) - 2 \pE \iprod{a_j, w} \iprod{a_\ell, w} \\
&\geq (\alpha' - O(\xi^4)) n \left( \sum_{i \in B_j} \| v_i \|^2 + \sum_{i \in B_k} \| v_k \|^2 \right) - O(\eta \xi^2) (\alpha')^2 n^2 \; .
\end{align*}
Observe that 
\begin{align*}
\sum_{i \in B_j} \| v_i \|^2 &= \sum_{i \in B_j, v_i \in B_j} \| v_i \|^2 + \sum_{ \in B_j, v_i \not\in B_j} \| v_i \|^2 \\
&\geq (1 - O(d \xi)) \alpha' n \left( \| \phi_0 \|^2  - d \xi \right) + O(d \xi) \alpha' n \\
&\geq \alpha' n \| \phi_0 \|^2 - O(m \xi) \alpha' n \; .
\end{align*}
since generically $\| v_i \|^2 = \pE w_i^2 \leq 1$.
Symmetrically we have $\sum_{k \in B_\ell} \| v_k \|^2 \geq (\| \phi_1 \|^2 - O(d \xi)) \alpha' n$.
Hence we have
\[
\sum_{i \in B_j, k \in B_\ell} \| v_i - v_k \|^2 &\geq (\| \phi_1 \|^2 + \| \phi_2 \|^2 - O(m \xi)) (\alpha')^2 n^2 \geq \Omega (\alpha')^2 \cdot (\alpha')^2 n^2 \; .
\]
Now suppose that $B_j \cap B_\ell \neq \emptyset$.
This implies that for all except for a $O(d \xi) (\alpha')^2 n^2$ set of pairs $i, j$ (i.e. those containing $v_i \not\in B_j$ or $v_j \not\in B_\ell$), the pairwise squared distance is at most $O(d \xi)$.
Since the pairwise distance between any two points is at most $2$, this is a clear contradiction.
\end{proof}

Finally, we show that a random point with large norm will likely be within a $B_\ell$.
\begin{lemma}
Let $i$ be a uniformly random index over the set of indices so that $\| v_i \|^2 \geq \alpha' / 100$.
Then, with probability $1 - O(d \xi)$, $v_i \in B_\ell$ for some $\ell$.
\end{lemma}
\begin{proof}
Observe that since $\| v_i \|^2 \leq 1$ and $\sum \| v_i \|^2 = \alpha' n$ there are at least $(1 - 1 / 100) \alpha' n$ vectors with $\| v_i \|^2 \geq \alpha' / 100$.
We have
\[
\sum_{\ell \not\in T} \| v_i \|^2 = \sum_{\ell \not\in T} \pE \iprod{a_\ell, w} \leq O(d \xi^2) \alpha' n \; ,
\]
so by Markov's inequality the number of $i$ with $i \in \cup_{\ell \not\in T} S_\ell$ and $\| v_i \|^2 \geq \alpha' / 100$ is at most $100 \cdot O(d \xi^2) n \ll O(m \xi) \alpha' n$.
There are at most $O(d \xi) \alpha' n$ vectors $v_i$ so that $v_i \in S_\ell$ for $\ell \in T$ and $v_i \not\in B_\ell$, and so the probability that a vector with $\| v_i \|^2 \geq \alpha' / 100$ is not of the desired form is at most $O(d \xi)$, as claimed.
\end{proof}

This completes the proof of Lemma~\ref{lem:rounding-main-nonuniform}, since this says that if we choose $i$ uniformly at random amongst all such $\| v_i \|^2 \geq \alpha / 100$, then with probability $1 - O(d \xi)$, we have $v_i \in B_\ell$ for some $B_\ell$ with $\| \phi_\ell \| = \Omega (\alpha')$, and hence if we look in a $O(\sqrt{d \xi})$ ball around it, it will contain all but a $O(d \xi) \alpha' n$ fraction of points from $S_\ell$.

\end{document}